\documentclass[11pt, reqno]{amsart}

\usepackage[utf8x]{inputenc}

\usepackage{amsmath, amsthm, amsfonts}
\usepackage{MnSymbol}
\usepackage{epsfig, enumerate}
\usepackage{color}
\usepackage{bbm,dsfont}

\DeclareMathOperator{\rank}{{\rm rank}}

\newtheorem{Theo}{Theorem}
\newtheorem{theorem}{Theorem}
\newtheorem{lemma}[theorem]{Lemma}

\newtheorem{prop}[theorem]{Proposition}

\newtheorem*{assumptions*}{Assumptions}
\newtheorem*{assumption}{Assumption}

\newtheorem*{rem*}{Remark}

\theoremstyle{remark}
\newtheorem{remark}[theorem]{Remark}

\theoremstyle{definition}
\newtheorem{definition}{Definition}

\newcommand{\bS}{{\mathbf S}}

\newcommand{\bfa}{{\mathbf a}}
\newcommand{\bfb}{{\mathbf b}}

\newcommand{\bfg}{{\mathbf g}}

\newcommand{\bfr}{{\mathbf r}}
\newcommand{\bfs}{{\mathbf s}}

\newcommand{\bfv}{{\mathbf v}}
\newcommand{\bfw}{{\mathbf w}}

\newcommand{\Aa}{{\mathcal A}}

\newcommand{\Dd}{{\mathcal D}}

\newcommand{\Gg}{{\mathcal G}}
\newcommand{\Hh}{{\mathcal H}}

\newcommand{\Oo}{{\mathcal O}}

\newcommand{\Vv}{{\mathcal V}}

\newcommand{\Ab}{{\mathbb A}}

\newcommand{\CC}{{\mathbb C}}
\newcommand{\DD}{{\mathbb D}}
\newcommand{\EE}{{\mathbb E}}

\newcommand{\GG}{{\mathbb G}}

\newcommand{\NN}{{\mathbb N}}

\newcommand{\PP}{{\mathbb P}}

\newcommand{\RR}{{\mathbb R}}
\newcommand{\Sb}{{\mathbb S}}

\newcommand{\VV}{{\mathbb V}}
\newcommand{\WW}{{\mathbb W}}

\newcommand{\ZZ}{{\mathbb Z}}

\newcommand{\Aaa}{{\mathfrak A}}
\newcommand{\BBb}{{\mathfrak B}}

\newcommand{\one}{{\bf 1}}
\newcommand{\nul}{{\bf 0}}

\newcommand{\qtx}[1]{\quad\text{#1}\quad}

\newcommand{\Mat}{{\rm Mat}}
\newcommand{\SL}{{\rm SL}}
\newcommand{\GL}{{\rm GL}}

\newcommand{\Her}{{\rm Her}}

\newcommand{\pmat}[1]{\begin{pmatrix} #1  \end{pmatrix}}
\newcommand{\smat}[1]{\left( \begin{smallmatrix} #1  \end{smallmatrix} \right)}

\newcommand{\diag}{{\rm diag}}
\DeclareMathOperator{\im}{{\rm Im}}
\DeclareMathOperator{\re}{{\rm Re}}

\DeclareMathOperator{\Tr}{{\rm Tr}}

\DeclareMathOperator{\supp}{{\rm supp}}
\DeclareMathOperator{\spec}{{\rm spec}}

\newcommand{\be}[1]{\begin{equation} \label{#1} }

\renewcommand{\CC}{\mathds{C}}
\renewcommand{\RR}{\mathds{R}}
\renewcommand{\ZZ}{\mathds{Z}}
\renewcommand{\NN}{\mathds{N}}

\numberwithin{theorem}{section}
\numberwithin{equation}{section}

\newcommand{\Ups}{\Upsilon}

\newcommand{\ovr}{\overrightarrow}

\usepackage{hyperref}

\title[Spectral Theory of one-channel Operators]{Spectral Theory of one-channel operators and application to absolutely continuous spectrum for Anderson type models}

\author{Christian Sadel}
\address{Facultad de Matem\'aitcas, Pontificia Universidad Cat\'olica de Chile} 
\email{chsadel@mat.uc.cl}


\begin{document}

\begin{abstract}
A one-channel operator is a self-adjoint operator on $\ell^2(\GG)$ for some countable set $\GG$ with a rank 1 transition structure along the sets of a quasi-spherical partition of $\GG$.
Jacobi operators are a very special case. In essence, there is only one channel through which waves can travel across the shells to infinity.
This channel can be described with transfer matrices which include scattering terms within the shells and connections to neighboring shells. Not all of the transfer matrices are defined for some countable set of energies. Still, many theorems from the world of Jacobi operators are translated to this setup. The results are then used to show absolutely continuous spectrum for the Anderson model on certain finite dimensional graphs with a one-channel structure. This result generalizes some previously obtained results on antitrees.

\end{abstract}

\maketitle


\section{Introduction and Results}

The main purpose of this paper is to discuss the probably most general frame work of Hermitian operators which allow a description through $2\times 2$ transfer matrices. We will call these 'one-channel operators'. 
In terms of classifying discrete Hermitian operators and their spectral theory we find this an interesting class.
Despite the one-channel structure,  the underlying graph-structures can have any kind of growth
and are not necessarily one-dimensional.
The work is inspired from \cite{Sa3}, but the considered operators are more general. 
Part of the motivation for this work comes from the theory of random disordered systems as considered in the Mathematical Physics community. There has quite been some interest in discrete (random) Schrödinger operators associated to certain discrete graph structures, e.g. 
\cite{ASW, AW, FHS, FHS2, KLW, KLW2, Kl, KS, MZ}.
In one-dimensional systems (Jacobi operators) transfer matrices are the main tool to study such operators and their behaviour can be connected to the spectrum and the spectral type in various ways \cite{CL, KP, KLS, LaSi}.
Due to transfer matrices, a local solution to the eigenvalue equation imposes a (formal) global one and a local analysis can lead to a full understanding of the spectral properties.
Here we generalize some of these methods to a more general class.

For general Schrödinger type operators on general graphs, the situation is typically much more difficult. This is part of the reason that the type of the bulk spectrum of the Anderson model in $\ZZ^d$ at small disorder is not understood.
It is a main open conjecture that for $d=2$ it should be pure point and for $d\geq 3$ there should be some absolutely continuous spectrum (almost surely).
Schrödinger operators associated to locally finite graph structures with hopping only along edges can always be brought into a block tri-diagonal matrix form by spherical-type partitions of the graph, cf. \cite{FHS}.
However,  the block sizes vary and typically grow. One still has resolvent identities, but generally the analysis can get very complicated.

\vspace{.2cm}

As an application of this work we obtain absolutely continuous spectrum for Anderson models on families of graphs of finite dimensional growth rate with dimension $d>2$.
We call these graphs partial antitrees as they partially have the structure as the antitrees in \cite{Sa3}. 
This is a first step of improving the ac-spectrum result in \cite{Sa3} towards more general classes of graphs with finite dimensional growth.

Previously to the work in \cite{Sa3} absolutely continuous spectrum for Anderson models with independent identically distributed potential had only been shown on infinite dimensional trees and tree-like graph structures 
\cite{Kl,ASW,FHS2,KS,KLW,AW,Sa,Sa2}. All these graph structures significantly simplify the general local Green's function relations as in \cite{FHS} and one has some simple local transfer mechanism to study them.
There has also been some work to approach the expected localization in two dimensions by certain graphs which in some sense are in between one and two dimensions \cite{MZ,Sa3}.

It is our agenda to further study certain hybrid graphs of lattice structures and antitrees and analyze what is happening with the absolutely continuous part of the spectrum for the Anderson models on these graphs.
We hope that this study can help to eventually make some progress in analyzing the important $\ZZ^d$-Anderson model towards showing existence of absolutely continuous spectrum.

Readers from the Mathematical Physics community should be aware that we will use a more mathematics-like notation in this paper: for instance $\lambda$ 
will be used as the spectral parameter (not $E$), and $\sigma$ will be a parameter describing the disorder strength (like $\sigma^2$ for variance).
But the results are actually non-perturbative in nature.

\subsection{One-channel operators}

Let $\GG$ be some countably infinite set and $\Hh$ a Hermitian operator on $\ell^2(\GG)$ with locally finite hopping, i.e. for all $x\in \GG$ the set $y\in \GG$ such that
$\langle \delta_x; \Hh \delta_y\rangle \neq 0$ is finite. Here, $\langle\cdot , \cdot \rangle$ denotes the standard scalar product 
with the convention that it is anti-linear in the first and linear in the second component, i.e. $\langle \psi, \phi \rangle = \sum_{x\in\GG} \bar \psi(x) \phi(x)$. With $\delta_x$ we denote the canonical basis vector with $\delta_x(x)=1$ and $\delta_x(y)=0$ for $x\neq y$.

\begin{definition}
 Let $\bS=(S_n)_{n\in\NN}$ or $\bS=(S_n)_{n\in\ZZ} $ be a partition of $\GG$ in finite sets, i.e. $\GG=\bigsqcup_n S_n$ and $0<\#(S_n)<\infty$. 
 The partition $\bS$ is called {\it quasi-spherical} for $\Hh$ if
 $$
 \langle \delta_x\;, \Hh\,\delta_y\rangle\,=\,0 \qtx{whenever} x\in S_m,\;y\in S_n \qtx{and} |n-m|\geq 2\;.
 $$
 In other words, $\Hh$ connects $S_n$ only to $S_{n-1}$, $S_n$ and $S_{n+1}$.
\end{definition}

Let us explain this notion and construct such a partition for operators with locally finite hopping. 
For $x\neq y$ we draw an edge between $x$ and $y$ if and only if $\langle \delta_x, \Hh \delta_y\rangle \neq 0$, then $\GG$ is a locally finite graph and $\Hh$ is associated to this graph structure.
Let $d(x,y)\in\NN$ denote the graph distance, i.e. the shortest path along edges connecting $x$ and $y$, $d(x,y)=\infty$ if they are on different edge-connected components.
Take some set $S_1$ which has at least one point of each connected component of $\GG$ and take $S_n:=\{x\in \GG: d(x,S_1)=\min_{y\in S_1} d(x,y)=n-1\}$, then $S_n$ is some kind of sphere around $S_1$ and $\bS=(S_n)_{n\in\NN}$ is such a quasi-spherical partition. 

\vspace{.2cm}

Now, let such a partition be given (not necessarily the above constructed one) and define $s_n=\#(S_n)$ to be the number of elements in $S_n$.
We can identify $\ell^2(\GG)=\bigoplus_n \ell^2(S_n) \cong \bigoplus_n \CC^{s_n}$ which gives $\ell^2(\GG)$ a fibered structure.
The direct sum is either running over $n\in \ZZ$ or $n\in \NN$ depending on the partition (an indexing over $\ZZ$ of this type is not always possible).

Vectors in $\CC^s$ will typically be understood as column vectors and the standard physics inner product can be written as $\langle \alpha, \beta \rangle_{\CC^{s}} = \alpha^* \beta$ where $\alpha^*$ is the transpose and complex conjugate of $\alpha\in\CC^s$ as a row vector.
We let $P_n$ be the canonical isometric injection from $\ell^2(S_n)$ into $\ell^2(\GG)$, 
then $P_n^*$ is the canonical projection, i.e.
$$P_n: \CC^{s_n}\cong \ell^2(S_n) \hookrightarrow \ell^2(\GG), \quad P_n^* : \ell^2(\GG)\to \ell^2(S_n)\cong \CC^{s_n} $$
and for $\Psi\in\ell^2(\GG)$ we define 
$$
\Psi_n:=P_n^*\Psi \in \ell^2(S_n)\cong \CC^{s_n}\;.
$$
Hence, for $\varphi_n\in \ell^2(S_n)$ and $\Psi\in \ell^2(\GG)$ one has 
$$P_n \varphi_n =\bigoplus_{k<n} \nul_{s_k} \oplus \varphi_n \oplus \bigoplus_{k>n} \nul_{s_k} \qtx{and} \Psi\,=\,\bigoplus_n P_n^* \Psi\,=\,\bigoplus_n \Psi_n \;.$$
Let us define 
$$
V_n:= P_n^* \Hh P_n\,\in\,\Her(s_n) \qtx{and} D_n:= P_{n}^* \Hh P_{n-1} \in \CC^{s_n \times s_{n-1}}
$$
where $\Her(s_n)$ denotes the set of Hermitian $s_n \times s_n$ matrices and $\CC^{s \times s'}$ is the set of complex $s \times s'$ matrices.
We find the tri-diagonal block structure of $\Hh$ as
\begin{equation}\label{eq-sph-form}
(\Hh \Psi)_n\,=\,D_{n+1}^* \Psi_{n+1}\,+\,D_{n} \Psi_{n-1}\,+\,V_n \Psi_n
\end{equation}
In the case where the partition is indexed over $\NN$ one sets $\psi_0=\nul$ and we may say that $\GG$ is partitioned as a 'half-space'.
The set $S_n\subset \GG$ may be referred to as the $n$-th shell of $\GG$.

\begin{definition} A self-adjoint operator $\Hh$ represented in the form \eqref{eq-sph-form} using a quasi-spherical partition $\bS$
is called a {\bf one-channel operator} with a channel across the partition $\bS$ if
{\bf for all $n$} one has $\rank D_n = 1$. Analogously, we call it an {\bf $r$-channel operator} if {\bf for all $n$} one has $\rank D_n=r$.
\end{definition}

Note, that if $\rank(D_{n_k})=1$ for a sequence $n_k \to \pm\infty$ with $k\to \pm \infty$, then one can get the one-channel property by grouping and taking $(S'_k)_k$ with $S_k'=\bigcup_{n_{k-1}<n\leq n_k} S_n$ as the new partition.
For people familiar with Dirac notations, for $\varphi_n\in \ell^2(S_n)$ the $s_n\times s_m$ matrix
$\varphi_n\varphi_m^*$ mapping from $\ell^2(S_m)$ to $\ell^2(S_n)$ simply represents the operator
$|\varphi_n\rangle \langle\varphi_m|$.
As an operator
on $\ell^2(\GG)$, $|\varphi_n\rangle \langle\varphi_m|$ can be represented as $\Psi \mapsto \,(\varphi_m^*\Psi_m)\,P_n \varphi_n $ or 
$P_n\varphi_n \varphi_m^* P_m^*$. In this work we will keep to the notations of matrix products of row with column vectors like $\varphi_n \varphi_m^*$ and we will not use the Dirac notations.

We assume now that $\Hh$ is a one-channel operator. We will show that its spectral properties can be described by transfer matrices.
Using a polar (Hilbert-Schmidt) decomposition of the $s_n \times s_{n-1}$ matrices $D_n$ we may write $D_n = -a_n\,\Ups_n \Phi_{n-1}^*$
where $\Phi_n\,,\,\Ups_n\in \ell^2(S_n)\cong \CC^{s_n}$ are {\bf unit} vectors and $0\neq a_n\in\RR$.
Then, we obtain
\begin{equation} \label{eq-Hh}
(\Hh\Psi)_n\,=\,-a_{n+1}\;\Phi_n\,\Ups_{n+1}^*\Psi_{n+1}\,-\,a_n\,\Ups_n \Phi_{n-1}^*\Psi_{n-1} \,+\,V_n\,\Psi_n\,,
\end{equation}
The minus sign is some convention to fit with notations of 1D-Schr\"odinger operators consisting of the negative Laplacian on $\ZZ$ plus a potential.
This makes later calculations more similar to \cite{CL}. If one does not like this convention one might always substitute $a_n=-\hat a_n$.

The direction of $\Phi_n$ is the {\bf forward going} mode in the $n$-th {\bf shell} $S_n$, connecting to $S_{n+1}$. Similarly, the direction $\Ups_n$ is the {\bf backward going} mode in the $n$-th shell.
If $\Ups_n=c_n \Phi_n$ for some number $c_n\neq 0$, then it is a {\bf propagating} mode through the $n$-th shell (cf. {\rm \cite{Sa3}}).

In the half-space case we take Dirichlet type boundary conditions at $0$, i.e. formally $\Psi_{0}=\Phi_{0}=0$.
In \cite{Sa3} we considered the special case when $\Phi_n=\Ups_n$ for all $n$ and called $\Hh$ an operator with
{\it one propagating channel}. Here, we do not assume this propagating property anymore. 
In the future, the half-space case will be included by considering the restriction of $\Hh$ to the half spaces, and we will assume that $n$ runs through all $\ZZ$. Thus, we define
\begin{equation}\label{eq-Hh+}
\GG^+=\bigsqcup_{n=1}^\infty S_n\;,\quad
\GG^-=\bigsqcup_{n=-\infty}^{0} S_n\;,\quad\Hh^+\,:=\,P_+^*\Hh P_+,\quad \Hh^-\,:=\,P_-^*\Hh P_-
\end{equation}
where $P_{\pm}$ are the canonical injections from $\ell^2(\GG^\pm)$ into $\ell^2(\GG)$, i.e. $\Hh^\pm$ are the restrictions of $\Hh$ to $\ell^2(\GG^\pm)$ with Dirichlet boundary conditions.

\vspace{.2cm}

Let us give some examples with this structure. Let $V_n$ be the adjacency matrix of some graph with vertex set $S_n$, let 
$\Phi_n=\delta_{x_n}, x_n\in S_n$ and for some subset $S'_n\subset S_n$ let $a_n \Ups_n=-\sum_{y\in S_n'}\delta_{y}$. 
Then, $\Hh$ is simply the adjacency operator of a graph consisting of the graphs $S_n$ glued together by connecting $x_n$ with every point in $S_{n+1}'\subset S_n$.

The ordinary graph Laplacian with some potential on the Sierpienski lattice as studied in \cite{MZ} can be represented as a 2-channel operator.

\begin{rem*}
 In principle one may want to allow the shells $S_n$ to contain infinitely many points. However, in this case many parts of the analysis in this paper only work to some extend.
 If $\#(S_n)=\infty$ then $V_n$ is an Hermitian operator on an infinite dimensional Hilbert space. In intervals $I\subset \RR$ where the spectrum of all the $V_n$ is discrete (no essential spectrum) all the techniques and theorems 
 of this paper would work just fine,  but in the cases of continuous or dense point spectrum of some of the $V_n$ many calculations would have to be adjusted.
 This will be dealt with elsewhere.
\end{rem*}

\subsection{A sufficient condition for unique self adjointness}

Let us get to the question of self-adjointness.
Recall that $\Hh^+$ is the operator $\Hh$ restricted to $\ell^2(\GG^+)=\bigoplus_{n=1}^\infty \ell^2(S_n)$ and similarly, 
$\Hh^-$ is $\Hh$ restricted to $\ell^2(\GG^-)=\bigoplus_{n=-\infty}^{0} \ell^2(S_n)$.
By $\Hh^\pm_{\min}$ we denote the operator $\Hh^\pm$ with minimal domain 
$$
\Dd^\flat_{\min}\,:=\,\left\{ \Psi\,=\bigoplus_n\,\Psi_n\,\in\,\ell^2(\GG^\flat)\,:\,\Psi_n=0\quad\text{for all but finitely many $n$}\;\; \right\}\;.
$$
where $\flat$ either symbolizes $+$ or $-$ or no index at all.
\begin{definition}
We say that $\Hh^\flat$ is {\it uniquely} self-adjoint, if $\Hh^\flat_{\min}$ with domain $\Dd^\flat_{\min}$ is essentially self-adjoint. 
\end{definition}
In this case, $\Dd^\flat_{\min}$ is a core and the self-adjoint closure has the maximal possible domain
$$
\Dd^\flat_{\max}\,:=\,\left\{\,\Psi\,\in\,\ell^2(\GG^\flat)\,:\, \Hh^\flat\,\Psi\,\in\,\ell^2(\GG)\;\right\}\;,
$$
We denote $\Hh^\flat$ with domain $\Dd^\flat_{\max}$ by $\Hh^+_{\max}$. It is a simple calculation that one always has $(\Hh^\flat_{\min})^*\,=\,\Hh^\flat_{\max}$.
The question of self-adjointness will be considered in more detail in Section~\ref{sec:selfadj}. It seems though that the famous limit-point criterion does not translate completely to this more general situation.
For now let us give some simple sufficient criterion.

\begin{prop} We find: \label{prop-unique-sa}
\begin{enumerate}[{\rm (i)}]
 \item If $\sum_{n=2}^\infty |a_n|^{-1}=\infty$, then $\Hh^+$ is uniquely self-adjoint.
 \item If $\sum_{n=0}^\infty |a_{-n}|^{-1}=\infty$, then $\Hh^-$ is uniquely self-adjoint.
 \item If both conditions, {\rm (i)} and {\rm (ii)} hold, then $\Hh$ is uniquely self-adjoint.
 \end{enumerate}
\end{prop}

\subsection{Transfer matrices}

The main tool for the fine spectral properties will be the transfer matrices which in this general setup are more complicated than in the Jacobian case.

Consider the difference equation $\Hh \Psi = z \Psi$ for any complex parameter $z$. 
Defining
\begin{equation}\label{eq-def-u}
 x_n\,=\,x_n(\Psi)\,:=\,\Ups_n^*\,\Psi_n\;,\qquad \tilde x_n\,=\,\tilde x_n(\Psi)\,:=\,\Phi_n^* \Psi_n
\end{equation}
we find after some trivial algebra that
\begin{equation}\label{eq-eig-1}
 (\Hh\Psi)_n=z\Psi_n \quad\Leftrightarrow \quad
 \Psi_n\,=\,(V_n-z)^{-1}\left( \Phi_n a_{n+1} x_{n+1}\,+\, \Ups_n a_n \tilde x_{n-1}\right)
\end{equation}
Next let us define
\begin{equation}\label{eq-def-chi_z}
 \Ups_{z,n}:=(V_n-z)^{-1} \Ups_n \qtx{and} \Phi_{z,n}:=(V_n-z)^{-1} \Phi_n
\end{equation}
as well as 
\begin{equation}\label{eq-def-alph}
 \alpha_{z,n}\,:=\,\Ups_n^* \Ups_{z,n}\,,\quad
 \beta_{z,n}\,:=\,\Ups_n^* \Phi_{z,n}\,,\quad
 \gamma_{z,n}\,:=\,\Phi_n^* \Ups_{z,n}\,,\quad
 \delta_{z,n}\,:=\,\Phi_n^*\Phi_{z,n}\;
\end{equation}
and
\begin{equation}\label{eq-def-bfg}
 \bfg_{z,n}\,:=\,\pmat{\alpha_{z,n} & \beta_{z,n} \\ \gamma_{z,n} & \delta_{z,n}}\,=\, 
 \pmat{\Ups_n^* \\ \Phi_n^*} (V_n-z)^{-1} \pmat{\Ups_n & \Phi_n}\;. 
\end{equation}
Then we get from \eqref{eq-eig-1} that
\begin{align}
 \pmat{ x_n \\ \tilde x_n}\,&=\,\pmat{ \beta_{z,n} & \alpha_{z,n} \\ \delta_{z,n} & \gamma_{z,n}} \pmat{a_{n+1}\,x_{n+1} \\ a_n  \tilde x_{n-1}} \,=\,\bfg_{z,n} \pmat{a_n  \tilde x_{n-1} \\ a_{n+1}\, x_{n+1}}\,.
\end{align}
For $\im(z)>0$ the matrix $\bfg_{z,n}$ is invertible unless $\Phi_n$ and $\Ups_n$ are co-linear.
If $\beta_{z,n}$ is invertible, then one can rearrange these equations to get $x_{n+1}$ and $\tilde x_n$ in terms of $x_n$ and $\tilde x_{n-1}$ as some kind of transfer equation:
$$
a_{n+1} x_{n+1}\,=\,\beta_{z,n}^{-1}(x_n-\alpha_{z,n} a_n \tilde x_{n-1})
$$
$$
\tilde x_n \,=\,\delta_{z,n} \beta_{z,n}^{-1}(x_n-\alpha_{z,n} a_n \tilde x_{n-1})\,-\,\gamma_{z,n} a_n\tilde x_{n-1}
$$
Therefore, we make the following definition.

\begin{definition}
We say that the channel is {\it broken} at the $n$-th shell $S_n$ if the function $z\mapsto \beta_{z,n} = \Ups_n^* (V_n-z)^{-1} \Phi_n$ is identically zero.
In this case the operator $\Hh$ can be splitted into a direct sum of two operators through the $n$-th shell $S_n$.
\end{definition}

We will now consider the case of a non broken channel, note that given the channel $(\Phi_n,\Ups_{n+1})_n$ this is a condition on $(V_n)_n$.
Then for energies where $\beta_{z,n}\neq 0$ the equations above can be written as
\begin{equation}\label{eq-transfer}
\pmat{a_{n+1} x_{n+1} \\ \tilde x_n}\,=\, T_{z,n} \pmat{a_n x_n \\ \tilde x_{n-1}}\qtx{where} T_{z,n}:=\pmat{a_n^{-1} \beta_{z,n}^{-1} & -a_n \beta^{-1}_{z,n}\alpha_{z,n} \\ a_n^{-1} \delta_{z,n} \beta_{z,n}^{-1} & a_n(\gamma_{z,n}- \delta_{z,n} \beta_{z,n}^{-1} \alpha_{z,n} )  }\,.
\end{equation}

For real values $z=\lambda\in\RR$ it is not hard to see that $\gamma_{\lambda,n}=\bar \beta_{\lambda,n}$ and $\alpha_{\lambda,n}$ and $\delta_{\lambda,n}$ are real. In particular, all entries of $\beta_{\lambda,n} T_{\lambda,n}$ are real
and $\frac{\beta_{\lambda,n}}{|\beta_{\lambda,n}|}\,T_{\lambda,n}\,\in\,\SL(2,\RR)$ and hence\footnote{$\exp(i\RR)\SL(2,\RR)$ denotes the set of $2\times 2$ matrices $T$ which can be written as
$T=e^{ir} T'$ with $r\in\RR$ and $T'\in\SL(2,\RR)$.}
$T_{\lambda,n}\in\exp(i\RR)\SL(2,\RR)$.

Note that in the special case $\Phi_n=\Ups_n$ we obtain $\alpha_{z,n}=\beta_{z,n}=\gamma_{z,n}=\delta_{z,n}$ and
$$
T_{z,n}\,=\,\pmat{a_n^{-1} \left(\Phi_n^*(V_n-z)^{-1} \Phi_n\right)^{-1} & -a_n \\ a_n^{-1} & 0 }
$$

The matrix $T_{z,n}$ is called the transfer matrix at level $n$ and spectral parameter $z$. 
The function $z\mapsto \beta_{z,n} T_{z,n}$ is well defined for $z$ in the upper half plane and hence $T_{z,n}$ 
is defined for $\im(z)>0$ except for possibly finitely many values of $z$ where $\beta_{z,n}=0$. 
Moreover, for these values, $T_{z,n}$ is invertible except for the finitely many values where $\gamma_{z,n}=\bar \beta_{\bar z,n}=0$.
We also define $T_{z,n}$ for real energies $z=\lambda \in \spec(V_n)$ where $z\mapsto T_{z,n}$ extends in a holomorphic way.

\begin{definition} Let us define some important subspaces of $\ell^2(S_n)\cong \CC^{s_n}$ for further analysis:
\begin{itemize}
 \item Let $\WW_n$ be the smallest subspace of $\ell^2(S_n)$ which is invariant under $V_n$ and contains $\Ups_n$. 
 \item Let $\tilde \WW_n$ be the smallest $V_n$-invariant subspace  of $\ell^2(S_n)$ containing $\Phi_n$.
 \item Let $\VV_n=\WW_n+ \tilde \WW_n$ as vector space sum, i.e. $\VV_n$ is the linear span of $\WW_n$ and $\tilde \WW_n$. 
\end{itemize}
\end{definition}

The following is a simple observation.
\begin{prop}\label{prop-eigf}
$\Hh$ and $\Hh^+$ leave the spaces $P_n (\VV_n^\perp) \cong \VV_n^\perp$ invariant where $n\in \ZZ$ or $n\in \NN$, respectively.
Moreover, $\Hh^+$ leaves $P_1(\tilde \WW_1^\perp)$ invariant\footnote{$\Hh^+$ does not depend on $\Ups_1$.}
Hence, $\Hh$ leaves also $\VV:=\bigoplus_{n=-\infty}^\infty \VV_n$ and $\Hh^+$ leaves $\VV^+:=\tilde\WW_1\oplus \bigoplus_{n=2}^\infty \VV_n$ invariant and we have
$$
\Hh \;\cong\; \Hh\,|\,\VV\;\oplus\;\bigoplus_{n=-\infty}^\infty V_n\,|\,\VV_n^\perp \quad,
\Hh^+ \;\cong\; \Hh^+\,|\,\VV^+\;\oplus\;V_1\,|\,\tilde \WW_1\,\oplus\,\bigoplus_{n=2}^\infty V_n\,|\,\VV_n^\perp \;.
$$
The right hand sides denote an orthogonal direct sum of self-adjoint operators and $V_n|\VV_n^\perp$ are ordinary matrices. Therefore, it is sufficient to study the spectral properties of 
$\Hh|\VV$ or $\Hh^+|\VV^+$ and one may assume $\VV_n=\CC^{s_n}$.
\end{prop}

Let us now classify those eigenvalues of $V_n|\VV_n$ for which one can still define the transfer matrix $T_{z,n}$.
The multiplicity of an eigenvalue of $V_n|\VV_n$ is at most 2 (cyclic space from 2 vectors). For defining $T_{z,n}$ the eigenvalues of multiplicity 2 are problematic as well as those of multiplicity one 
which are not eigenvalues of $V_n|(\WW_n\cap\tilde \WW_n) $. Together, these problematic eigenvalues are exactly the eigenvlaues on the quotient, $\spec(V_n| \VV_n\,/\,(\WW_n\cap \tilde \WW_n))$. 
Note that in the particular case where $\WW_n=\tilde \WW_n$ this set is empty.

The non-problematic eigenvalues of $V_n|\VV$ are those of multiplicity one with eigenvector in $\WW_n\cap\tilde \WW_n$, in other words, exactly the eigenvalues of $\spec(V_n|\WW_n\cap\tilde \WW_n)$.
\begin{prop}\label{prop-def-T}
$T_{z,n}$ is not defined if $\beta_{z,n}= 0$ or if $z=\lambda\in \spec(V_n|\VV_n\,/\,(\WW_n\cap\tilde \WW_n))$.
$T_{z,n}$ is not invertible if $\gamma_{z,n}=0$.
In all other cases, in particular if $\lambda\in\spec(V_n|\WW_n\cap \tilde \WW_n)$, the transfer matrix $T_{\lambda,n}$ is well defined (by holomorphic extension) and invertible.
In these cases also the vector $\Psi_{\lambda,n}^{x}$ (as defined in \eqref{eq-def-psi_z} below)  associated to a solution $(x,\tilde x)$ of the transfer matrix equation 
is well defined.
\end{prop}

After these extensions we let 
\begin{equation}
A_n\,:=\{z\in\CC\,:\, T_{z,n} \;\;\text{or}\;\; T_{z,n}^{-1}\;\;\text{is not defined}\; \}\;.
\end{equation}
Thus $A_n=\{z\,:\,\beta_{z,n}\gamma_{z,n}=0\} \cup \,\spec(V_n|\VV_n\,/\,(\WW_n\cap \tilde \WW_n)) $.
\begin{equation}
B_{l,m}\,:=\,\bigcup_{l\leq n \leq m} A_n\;, \; B_\infty:= B_{-\infty,\infty}=\bigcup_{n=-\infty}^\infty A_n
\end{equation}
where we allow $l=-\infty$ or $m=\infty$.
Furthermore, for $l\leq m$; $l,m\in\ZZ$, $z\not \in B_{l+1,m}$ we define 
\begin{equation} 
T_{z,l,m}\,:=\,T_{z,m} T_{z,m-1}\;\cdots\;T_{z,l+1},\quad T_{z,l,l}:=\one  \qtx{and}
T_{z,m,l}\,:=\,(T_{z,l,m})^{-1}\;.
\end{equation}

One should think of $\ovr{x}_n :=\smat{a_{n+1} x_{n+1} \\ \tilde x_n}$ for a solution as the $n$-th vector and $T_{z,l,m}$ maps $\ovr{x}_l$ to $\ovr{x}_m$ and 
we have the special one-step transfer matrices $T_{z,n}=T_{z,n-1,n}$.
If we do not have a broken channel, then except for countably many energies, all transfer matrices are defined and invertible.
At each $n$, $T_{z,n}$ is not defined for only finitely many energies.
More precisely, $\det(V_n-z) \beta_{z,n}=0$ leads to a polynomial equation of degree $n-1$ so it has $n-1$ solutions in $z$ counted with multiplicity.
Also, $\gamma_{\bar z,n}=\bar \beta_{z,n}$ so the energies where $T_{z,n}$ is not invertible relate to complex conjugate energies where $T_{z,n}$ is not defined.
The used indices and induced index shifts can be sometimes confusing, therefore it is good to keep the following schematic picture in mind:

\vspace{.2cm}

\begin{center}
\includegraphics[width=.9\linewidth]{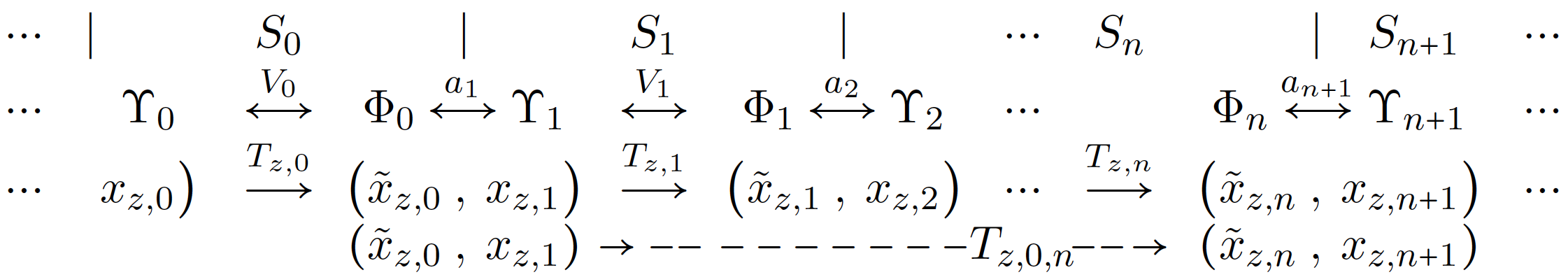} 
\end{center}

\vspace{.2cm}

Any pair of sequences $(x_{n+1},\tilde x_n)_n$ solving the above equation \eqref{eq-transfer} for some interval $n\in[m,l]\cap \ZZ$ will be called a solution of the transfer matrix equation at $z$.
Such a solution will simply be abbreviated by $(x,\tilde x)$ or simply by $x$. Note that the sequence $x_n$ actually also defines the sequence $\tilde x_n$ by \eqref{eq-transfer}.
For any such solution we let
\begin{equation}\label{eq-def-psi_z}
 \Phi_{z,n}^{x}\,:=\, a_{n+1} x_{n+1}\,\Phi_{z,n}\;, \quad
 \Ups_{z,n}^{x}\,:=\,a_n \tilde x_{n-1}\,\Ups_{z,n}\;, \quad
 \Psi_{z,n}^{x}\,:=\, \Phi_{z,n}^{x}\,+\,\Ups_{z,n}^{x}
\end{equation}
so that $\bigoplus_n \Psi_{z,n}^{x}$ is the (formal) solution of the eigenvalue equation.

We will work with some special solutions and let $(u_z, \tilde u_z)$ and $(w_z, \tilde w_z)$ be the solutions at $z$ with $u_{z,1}=1=\tilde w_{z,0} a_1$ and $\tilde u_{z,0}=0=w_{z,1}$. Then
\begin{equation}\label{eq-T-u-v}
 T_{z,0,n}\,=\,\pmat{a_{n+1} u_{z,n+1} / a_1 & a_{n+1} w_{z,n+1} a_1 \\ \tilde u_{z,n} / a_1 & \tilde w_{z,n} a_1}
 \,=\,\pmat{a_{n+1} u_{z,n+1} & a_{n+1} w_{z,n+1} \\ \tilde u_{z,n}  & \tilde w_{z,n} }\pmat{a_1^{-1} \\ & a_1}
\end{equation}
and we will use
$\Phi_{z,n}^u,\; \Ups_{z,n}^u,\;\Psi_{z,n}^u$ instead of $\Phi_{z,n}^{u_z},\; \Ups_{z,n}^{u_z},\; \Psi_{z,n}^{u_z}$, respectively.
This particular choice of solutions mimics the explicit dependence of the transfer matrices on $a_1$ for $n\geq 1$. 
Since $\Hh^+$ does not depend on $a_1$ this choice makes sense for considering $\Hh^+$.

\subsection{Spectral Theory}

From now we will always assume that $\Hh$ and hence also $\Hh^\pm$ are uniquely self-adjoint.

Besides the eigenfunctions of $V_n|\VV_n^\perp$, let us first classify all further eigenfunctions with finite support.
Such eigenfunctions can be created on shells between non-existent transfer matrices.
First, if for a real value $\lambda\in\RR$ the transfer matrix $T_{\lambda,n}$ does not exist, then we associate 'boundary conditions' on the left and right of $S_n$ (cf. Lemma~\ref{lem-boundary}).
Therefore, we define for $\lambda\in A_n\cap\RR$
\begin{equation}
 \ovr{x}_{\lambda,n}^{(-)}\,:=\,\pmat{a_{n+1}^2 \alpha_{\lambda,n} \\ 1} \qtx{or} \ovr{x}_{\lambda,n}^{(-)}\,:=\,\pmat{1 \\ 0} \quad \text{if $|\alpha_{\lambda,n}|=\infty$}
\end{equation}
and
\begin{equation}
 \ovr{x}_{\lambda,n}^{(+)}\,:=\,\pmat{1 \\ \delta_{\lambda,n}} \qtx{or} \ovr{x}_{\lambda,n}^{(+)}\,:=\,\pmat{0 \\ 1} \quad \text{if $|\delta_{\lambda,n}|=\infty$.}
\end{equation}
In these definitions, $\alpha_{\lambda,n}=\Ups_n^*(V_n-\lambda)^{-1} \Ups_n$ exists if the restriction $V_n-\lambda|\WW_n$ is invertible, otherwise $|\alpha_{\lambda,n}|=\infty$.
Similarly,  $\delta_{\lambda,n}$ exists if and only if $V_n-\lambda|\tilde\WW_n$ is invertible, otherwise $|\delta_{\lambda,n}|=\infty$.

\begin{Theo} \label{th-compact-eig} Let $\lambda\in\RR$ and let the channel not be broken somewhere. Then, we have the following:
\begin{enumerate}[{\rm (i)}]
\item Suppose $l\leq m$, $\lambda\not\in B_{l+1,m}$, $\lambda\in A_{l}\cap A_{m+1}$ and $T_{\lambda,l,m} \ovr{x}^{(+)}_{\lambda,l}=C \ovr{x}^{(-)}_{\lambda,m+1}$ for some constant $C\in\CC$.
Then, $\Hh|\VV$ has precisely one eigenfunction (up to scalars) with eigenvalue $\lambda$ which is supported on $\bigcup_{n=l-1}^{m+1} S_n$.
For $l>0$ this is also an eigenfunction of $\Hh^+$.
\item Suppose $T_{\lambda,n}$ exists for $1\leq n \leq m$ and not for $n=m+1$ and $T_{\lambda,0,m} \smat{1\\0}=C \ovr{x}^{(-)}_{\lambda,m+1}$. 
Then, $\Hh^+|\VV^+$ has precisely one eigenfunction (up to scalars) with eigenvalue $\lambda$ which is supported on $\bigcup_{n=1}^{m+1} S_n$.
\item Any eigenfunction of $\Hh|\VV$ or $\Hh^+|\VV^+$ with compact support is a finite sum of such eigenvectors as in {\rm (i)} or {\rm (ii)}.
\item Let $\lambda \in A_l$, $\lambda\not \in A_n$ for all $n>l$ and let
$\smat{a_{n+1} x_{n+1} \\ \tilde x_n}=T_{\lambda,l,n}\ovr{x}^{(+)}_{\lambda,l}$, be the corresponding solution for $l<n$.
Then, for some specific $\psi_l\in\ell^2(S_l)$ the vector $\psi_l\oplus \bigoplus_{n>l} \Psi^x_{\lambda,n}$ is a formal solution of $\Hh\Psi=\lambda\Psi$.
If it is in $\ell^2(\GG)$, then it is the unique eigenfunction (up to scalars) of $\Hh|\VV$ (and $\Hh^+|\VV^+$ if $l> 0$) with eigenvalue $\lambda$ which is supported on $\bigsqcup_{n=l}^\infty S_n$.
\item Let $\lambda \in A_{m+1}$, $\lambda\not \in A_n$ for all $n\leq m$ and let $\smat{a_{n+1} x_{n+1} \\ \tilde x_n}=T_{\lambda,m,n}\ovr{x}^{(-)}_{\lambda,m+1}$, be the corresponding solution for $n<m$. 
Then, for some specific $\psi_{m+1}\in\ell^2(S_{m+1})$ the vector  $\bigoplus_{n\leq m} \Psi^x_{\lambda,n} \,\oplus\,\psi_{m+1}$ is a formal solution of $\Hh\Psi=\lambda\Psi$.
If it is in $\ell^2(\GG)$, then it is the unique eigenfunction (up to scalars) of $\Hh|\VV$ (and $\Hh^+|\VV^+$ if $l> 0$) with eigenvalue $\lambda$ which is supported on $\bigsqcup_{n=-\infty}^{m+1} S_n$.
\end{enumerate}
\end{Theo}

\begin{rem*}
 In case (ii), the special limiting case where $m=0$ and $T_{\lambda,1}$ is not defined with $\ovr{x}^{(-)}_{\lambda,1}=\smat{1\\0}$, i.e. $\Phi_1^*(V_1-\lambda)^{-1} \Phi_1=0$ leads to an eigenvector
 $\varphi$ of $\WW_1 \cap \tilde \WW_1^\perp$ which is formally already included in Proposition~\ref{prop-eigf}.
 The existence of $T_{\lambda,1}$ depends on the $\Ups_1$ but changing this vector does not change the operator $\Hh^+$, it however changes the space $\VV_1$. If an eigenvector of $\Hh^+$ for $\lambda$ lies in $\VV_1^\top$ it does not affect existence of 
 $T_{\lambda,1}$, but if it lies in $\VV_1 \setminus \tilde \WW_1$, then $T_{\lambda,1}$ does not exist.\\
 \end{rem*}

Thus, real values were certain transfer matrices do not exist can lead to compactly supported eigenfunctions. Complex energies where certain transfer matrices do not exist could be seen as some form of resonances but it might be interesting to classify these further.

Let us now  define the spectral measures $\mu_{\Psi}$ and $\mu_{\Psi}^{\pm}$ by
\begin{equation}
\langle \Psi; f(\Hh^\flat) \Psi \rangle\,=\,\int_\RR f(\lambda)\,\mu^\flat_{\Psi}(d\lambda)
\end{equation}
where $\flat$ either denotes $+$, $-$ or no upper index.
For simplicity we may write $\mu^\flat_{\Ups_n}$ and $\mu^\flat_{\Phi_n}$ instead  of the formally correct notation $\mu^\flat_{P_n \Ups_n}$ or $\mu^\flat_{P_n\Phi_n}$.

\begin{Theo}\label{th:H+msr}
Let $\Hh$ and hence also $\Hh^+$ be uniquely self-adjoint and let the channel not be broken.
\begin{enumerate}[{\rm (i)}]
\item There is some pure point measure $\nu^+$ supported on $B_{1,\infty}\cap \RR$ and induced by the eigenvectors as in Theorem~\ref{th-compact-eig}, such that
 \begin{align}\label{eq-mu1+}
&a_1^2\,{\rm d}\mu_{\Ups_1}^+(\lambda)\,=\,{\rm d}\nu^+(\lambda)\,+\,\lim_{n\to\infty} \pi^{-1}\,\left\| T_{\lambda,0,n} \smat{1 \\ 0} \right\|^{-2}{\rm d}\lambda\;
\end{align}
where ${\rm d}\lambda$ denotes the ordinary Lebesgue measure along $\RR$ and the limit is a weak limit of measures.
\item For $\varphi \in \VV_n\subset \ell^2(S_n),\,n\in\NN$ we find that $\mu^+_{P_n\varphi}$ is absolutely continuous with respect to $\mu_{\Ups_1}^+$ on $\RR\setminus B_{1,n}$ and in the sense of an $\mu_{\Ups_1}$-almost sure defined Radon Nikodym derivative, 
\begin{align}
& \frac{{\rm d}\mu_{P_n \varphi}^+(\lambda)}{{\rm d}\mu^+_{\Ups_1}(\lambda)} \,=\, |\varphi^*\Psi_{\lambda,n}^u|^2\,,\quad \lambda\in \RR\setminus B_{1,n}\;.
\end{align}
\item There is some positive pure point measure $\nu$ supported on $B_{\infty}\cap \RR$ and induced by the eigenvectors as in Theorem~\ref{th-compact-eig}, such that
for any sequence of real unit vectors $\ovr{x}_m\in\RR^2$ we have
 \begin{align}
& {\rm d}\left(a_1^2\mu_{\Ups_1}\,+\,\mu_{\Phi_0}\,-\,\nu\right)(\lambda)\,=\, \notag \\
&\qquad =\, \lim_{m,n\to\infty} \frac{d\lambda}{\pi}\,\frac{\left\|T_{\lambda,-m,0} \ovr{x}_m \right\|^2}{\left\| T_{\lambda,-m,n}\ovr{x}_m \right\|^2}\,=\,
\lim_{m,n\to\infty} \frac{d\lambda}{\pi}\,\left\| T_{\lambda,0,n} \frac{T_{\lambda,-m,0} \ovr{x}_m}{\|T_{\lambda,-m,0} \ovr{x}_m\|} \right\|^{-2} \\
&\qquad =\,\lim_{m,n\to\infty} \frac{{\rm d}\lambda}{\pi^{2}}\,\int_0^\pi{\rm d\theta} \left\{
\left\| T_{\lambda,0,-m} \smat{\cos(\theta) \\ \sin(\theta)} \right\|^{-2}\left\| T_{\lambda,0,n} \smat{\cos(\theta) \\ \sin(\theta)} \right\|^{-2}\;\right\}\,.
\label{eq-mu1}
\end{align}
where ${\rm d}\lambda$ denotes the ordinary Lebesgue measure along $\RR$. The double limits $m,n\to\infty$ are weak limits of measures and can be taken along any sequence $n_k, m_k$.
\item For $\varphi \in \VV_n\subset \ell^2(S_n),\,n\in\ZZ$ we find that $\mu_{P_n\varphi}$ is absolutely continuous with respect to $\mu_{\Phi_0}+ \mu_{\Ups_1}$ on $\RR\setminus B_{\infty}$ and for the Radon-Nikodym derivative we find
\begin{align}
&\frac{{\rm d}\mu_{P_n \varphi}(\lambda)}{{\rm d}(\mu_{\Ups_1}+a_1^2 \mu_{\Phi_0})(\lambda)} \,\leq\, |\varphi^* \Psi^u_{\lambda,n}|^2\,+\,|\varphi^*\Psi^w_{\lambda,n}|^2
\end{align}
\end{enumerate}
\end{Theo}

\begin{rem*} 
\begin{enumerate}[{\rm (i)}]
\item Note that $\Hh^+$ does not depend on $\Ups_1$ and $a_1$, but $T_{z,1}$, $T_{z,0,n}$, $\nu^+(d\lambda)$ and also $A_1$, $B_{1,n}$ do.
Particularly, this means that one can take any $\Ups\not\in \tilde \WW_0^\perp$ (no broken channel) and apart from the eigenvalues as classified in Theorem~\ref{th-compact-eig}, 
the measure $\mu_{\Ups}^+$ gives the spectral type of $\Hh^+|\VV^+$.
For some choices of $\Ups_1=\Ups$ the measure $\nu^+$ has additional support on some points coming from eigenfunctions of $\Hh^+$ supported in $S_1$.
\item The measure $\mu_{\Phi_0}+\mu_{\Ups_1}$ gives the spectral type of $\Hh|\VV$ apart from the possible eigenvalues in $B_\infty$ with eigenfunctions as described in Theorem~\ref{th-compact-eig}.
\end{enumerate}
\end{rem*}

An immediate consequence is the following criteria for pure absolutely continuous spectrum which is well known in the Jacobi operator case \cite{CL,LaSi}.
It follows directly from the above Theorem with essentially the same proof as in \cite[Theorem~1.3]{LaSi}

\begin{Theo}\label{th:ac-spec}
Assume that for some $p>2$ and some interval $(a,b)$ one has 
$$
\liminf_{n\to \infty} \int_a^b \| T_{\lambda,0,n}\|^p\,d\lambda\,<\,\infty\,.
$$
Then, the interval $[a,b]$ is in the spectrum of $\Hh^+|\VV^+$ and of $\Hh|\VV$. Furthermore, the spectrum of $\Hh^+|\VV^+$ is purely absolutely continuous in $(a,b)\setminus B_{1,\infty}$ and the spectrum of $\Hh|\VV$ is purely absolutely continuous in $(a,b)\setminus B_\infty$.
Moreover, the density of $\mu^+_{\Ups_1}$ and $\mu_{\Ups_1}+\mu_{\Phi_0}$ restricted to $(a,b)\setminus B_\infty$ with respect to the Lebesgue measure is in $L^{p/2}((a,b))$.
Possible eigenfunctions with eigenvalues in $(a,b)\cap B_\infty$ are of the form as described in Theorem~\ref{th-compact-eig}.
\end{Theo}


\section{Application to random operators}

One may choose the data $\Xi_n:=(V_n,\;a_n,\;,\Phi_n,\;\Upsilon_n)$ as random variables depending on some element $\omega$ in an abstract probability space $\Omega$.
If we make sure to almost surely fulfill some unique self-adjointness criterion then we create
a random family of Hermitian one-channel operators $\Hh_\omega$ (and $\Hh^+_\omega$ in the half space), with respect to a fixed (non-random) partition $\bS$.

In this case, $\alpha_{z,n},\,\beta_{z,n},\,\gamma_{z,n},\,\delta_{z,n}$ are random variables. Moreover,  the sets $A_n$ and hence $B_{1,\infty},=B_{1,\infty}(\omega)$,  $B_\infty=B_\infty(\omega)$ are random as well as
the spaces $\VV=\VV_\omega$ and $\VV^+=\VV^+_\omega$.
If one chooses $(\Xi_n(\omega))_n$ to be stochastically independent (with respect to $n$), then for any fixed $z$, the vectors
transfer matrices will be stochastically independent in $n$ as well.

In order to make the statement a bit easier let us give the following definition. 

\begin{definition}\label{def-tn-fluc}
Let $(X_n)_{n\in \NN}$ be some sequence of real or complex valued random variables and let $(t_n)_{n\in\NN}$, $t_n>0$ be some sequence of positive numbers with $t_n\to 0$ for $n\to \infty$.
We say that {\it $X_n$ is well-balanced order $\Oo(t_n)$ close to $X\in \CC$ up to $K$ moments}, iff there are constants $C_k$, $k=0,1,\ldots K$ such that
$$
\EE(|X_n-X|^k)\,\leq\,C_k\,t_n^k\qtx{and} |\EE(X_n)-X|<C_0 t_n^2\;
$$
for all $n\in\ZZ_+$ and all $k=1,\ldots,K$.
 Here and below,  $\EE$ will always denote the expectation value of a random variable. Symbolically we will write
 $$
 X_n \xrightarrow[K]{\Oo(t_n)} X \qtx{and} X_n \xrightarrow{\Oo(t_n)} X
 $$
 if it holds for every $K\in\NN$.
 While $|X_n-X|$ can fluctuate on the order of $t_n$, the second, well-balanced condition states that the expectation $\EE(X_n)$ only fluctuates to the order of $t_n^2$ around $X$.
 \\
A family $(X_n(\lambda))_n$ depending on an additional parameter $\lambda \in I$ is {\it uniformly} well-balanced $\Oo(t_n)$ close to $X(\lambda)$ up to $K$ moments if one can choose the constants
$C_k(\lambda)$ as above to be uniformly bounded in $\lambda\in I$. 
Finally we make the same definition for sequences of vector valued random variables by considering each component individually.
\end{definition}

\begin{Theo}\label{th-random}
Let $\Hh_\omega$ be some random family of Hermitian one-channel operators (with respect to a fixed partition $\bS$) such that the transfer matrices $(T_{\lambda,n})_n$ 
are stochastically independent in $n$.
Let $I\subset \RR$ be some open set and for $\lambda\in I$ let\footnote{$\exp(i\RR) \SL(2,\RR)$ denotes the set of $2\times 2$ matrices which are given by the product of a unit element times a real matrix with determinant 1. For real $\lambda$ all transfer matrices $T_{\lambda,n}$ which exist are of this form.}   $T(\lambda) \in \exp(i\RR) \SL(2,\RR)$ be some continuous function and $(t_n)_n$ some sequence
such that
$$
\sum_{n=1}^\infty t_n^2 < \infty \qtx{and}
|\Tr(T(\lambda))|\,<\,2 \qtx{and} T_{\lambda,n} \xrightarrow[4]{\Oo(t_n)} T(\lambda) 
$$
uniformly on compact intervals $\lambda\in [a,b]\subset I$.\\
Then, the spectrum of $\Hh_\omega|\VV_\omega$ in $I\setminus B_\infty$ and the spectrum of $\Hh^+_\omega|\VV^+_\omega$
in $I\setminus B_{1,\infty}$ is almost surely 
purely absolutely continuous.
\end{Theo}

Now as an application we will consider the Anderson model on certain graphs which generalize the antitrees as considered in \cite{Sa3}.
For this we will focus on 'half-space' cases, this means that the index $n$ will now only go through $\NN$ and the random operator $\Hh_\omega$ corresponds to 
$\Hh^+$ above.

\subsection{Partial antitrees}

A partial antitree is a graph $\GG$ of the following structure: As a point set, 
the graph is the disjoint union of finite sets,
$$
\GG=\bigsqcup_{n=1}^\infty R_n, \qtx{where} r_n:=\#(R_n)<\infty
$$
denotes the number of vertices in $R_n$ and we assume that
$$
r_{3n-2}> 0, \quad r_{3n-1} \geq 0, \quad r_{3n}>0 \qtx{for} n\in\NN \;.
$$
So in particular we allow that some of the sets $R_{3n-1}$ for some $n\in \NN$ are empty, but the sets $R_{3n}$ and $R_{3n+1}$ are never empty.
Between the points $R_{3n}$ and $R_{3n+1}$ we draw every edge, i.e. every vertex of $R_{3n}$ is connected to each vertex of $R_{3n+1}$ and
the edge weights will be re-normalized. These are the antitree-type connections as in \cite{Sa3}. 
In order to describe them later let us define the unit vector
$$
\varphi_n:=\frac{1}{\sqrt{r_n}} \pmat{1\\\vdots \\1}\,\in\,\ell^2(R_n)\cong \CC^{r_n}\,.
$$
For the connections within a group 
$$S_n:=R_{3n-2}\sqcup R_{3n-1}\sqcup R_{3n},\quad{n\geq 1} $$
we just want $S_n$ to be a connected graph for a most general partial antitree. 

The partition $\bS=(S_n)_{n\in\NN}$ is a quasi-spherical partition. The $\#(S_n) \times \#(S_n)$ matrix $A_n$ shall describe the edges and weights of connections within $S_n$, i.e. $(S_n)_{ij}=0$ if there is no edge between $i,j \in S_n$ and $(S_n)_{ij}=(S_n)_{ji} \in \RR$ denotes the weight of an edge between $i,j \in S_n$ if there is one.
Then letting 
 $$
\Upsilon_n=\pmat{\varphi_{3n-2} \\ \nul \\ \nul}\,\quad \Phi_n=\pmat{\nul \\ \nul \\ \varphi_{3n}},\quad a_n=-1 \qtx{and} V_n=A_n\,
$$
the adjacency operator $\Aa$ of the partial antitree corresponds to $\Hh^+$ as in \eqref{eq-Hh} and \eqref{eq-Hh+}\footnote{One may note that the choice of $\Upsilon_1,\,a_1$ does not change the operator $\Hh^+$}.
We call this part of the adjacency operator $A_n$ as for the Anderson model later we will change to $V_n$ and add an additional random diagonal potential.

\begin{remark}
Note that if we set $R_{3n-1}=\emptyset$ for all $n$, hence, $S_n=R_{3n-2}\sqcup R_{3n}$, and connect each vertex of $R_{3n-2}$ with each one of $R_{3n}$ and normalize the weights by $1/\sqrt{r_{3n-2} r_{3n}}$
then $(\GG,\Aa)$ corresponds to an antitree with normalized edge weights as defined in \cite{Sa3} and 
$A_n=\smat{\nul & \varphi_{3n-2}\varphi_{3n}^* \\ \varphi_{3n} \varphi_{3n-2}^* & \nul}$.
\end{remark}

Here we will work with some examples with additional homogeneity structures and consider the following partial antitrees:

\vspace{.2cm}

\noindent {\bf Example 1:} The stretched antitree $\Sb_\bfs$ associated to the sequence $\bfs=(s_n)_n, s_n\in\NN$:
Here, $r_{3n-1}=0$, i.e. $R_{3n-1}=\emptyset$ for all $n\in \NN$, and  $r_{3n-2}=r_{3n}=s_n$ and
 $A_n=\smat{\nul & \one \\ \one & \nul}$ where each entry is an $s_n\times s_n$ block.
See Figure~\ref{fig1}
\vspace{.2cm}

\noindent {\bf Example 2:} A partial antitree with homogeneous connecting modes: 
Let $k_1, k_2, k_3$ be positive integers\footnote{we allow $k_2=0$ iff $r_{3n-1}=0$ for all $n$, otherwise $k_2>0$, we always will have $k_1>0,\,k_3>0$} and $k=k_1+k_2+k_3$.
We assume that $r_n=\#(R_n)$ is divisible by $k_i$ for $n\equiv j \mod 3,\;\; (j=1,2,3)$ and we split
$R_n=\bigsqcup_{i=1}^{k_j} R_{n,i}$ into sets of equal size,  i.e. $\#(R_{n,i})=r_n / k_j$.
We either assume $r_{3n-1}>0$ for all $n$ in which case $k_2>0$, or $r_{3n-1}=0$ for $n$ in which case we set $k_2=0$.
Then let $O\in {\rm O(k)}$ be an orthogonal\footnote{i.e. $O^*=O^\top=O^{-1}$}, 
$k\times k$ matrix and let $\bfa=\diag(a_1,\ldots,a_k)$ be a real diagonal matrix.
Then we let $O_n$ be a $(r_{3n-2}+r_{3n-1}+r_{3n}) \times k=\#(S_n) \times k$ matrix as operator from $\CC^k$ to $\ell^2(S_n)$ defined by 
$$
(O_n)_{im}=\begin{cases} \sqrt{k_j/r_{3n-2}} \;O_{\ell m}, &\qtx{if} i\in R_{3n-2,\ell},\quad\ell=1,\ldots,k_1; \\ 
\sqrt{k_j/r_{3n-1}} \;O_{\ell+k_1, m}, &\qtx{if} i\in R_{3n-1,\ell},\quad\ell=1,\ldots,k_2; \\
\sqrt{k_j/r_{3n}} \;O_{\ell+k_1+k_2, m}, &\qtx{if} i\in R_{3n-2,\ell},\quad\ell=1,\ldots,k_3;
\end{cases}
$$
Note that $O_n: \CC^k\to \ell^2(S_n)$ is an isometry, $O_n^* O_n=\one$.
Then we let $A_n:=O_n \bfa O_n^*$.
We denote this partial antitree by $\GG=\Ab_\bfr(\vec{k},O,\bfa)$ or short $\Ab_\bfr$, where $\bfr=(r_n)_n$ symbolizes the sequence $r_n$ and
$\vec{k}=(k_1,k_2,k_3)$ the vector $k$.
\vspace{.2cm}

The vectors in $O$ represent shapes of spherical waves. For knowing the graph structure and weighs it is sufficient to know $O\bfa O^*$.
Some special example below are the graphs $\widehat \Ab_\bfr$ with $k_1=k_2=k_3=2$ and 
$O\bfa O^*=\smat{\nul & \one & \nul \\ \one & \nul & \one \\ \nul & \one & \nul}$ where each entry is a $2\times 2$ matrix and $O\bfa O^*$ is a $6\times 6$ matrix. Some example is given in Figure~\ref{fig2}.

\begin{figure}[ht]
\begin{center}
\begin{tabular}{cc}
\begin{minipage}{6.4cm}
\begin{center}
\includegraphics[width=5cm]{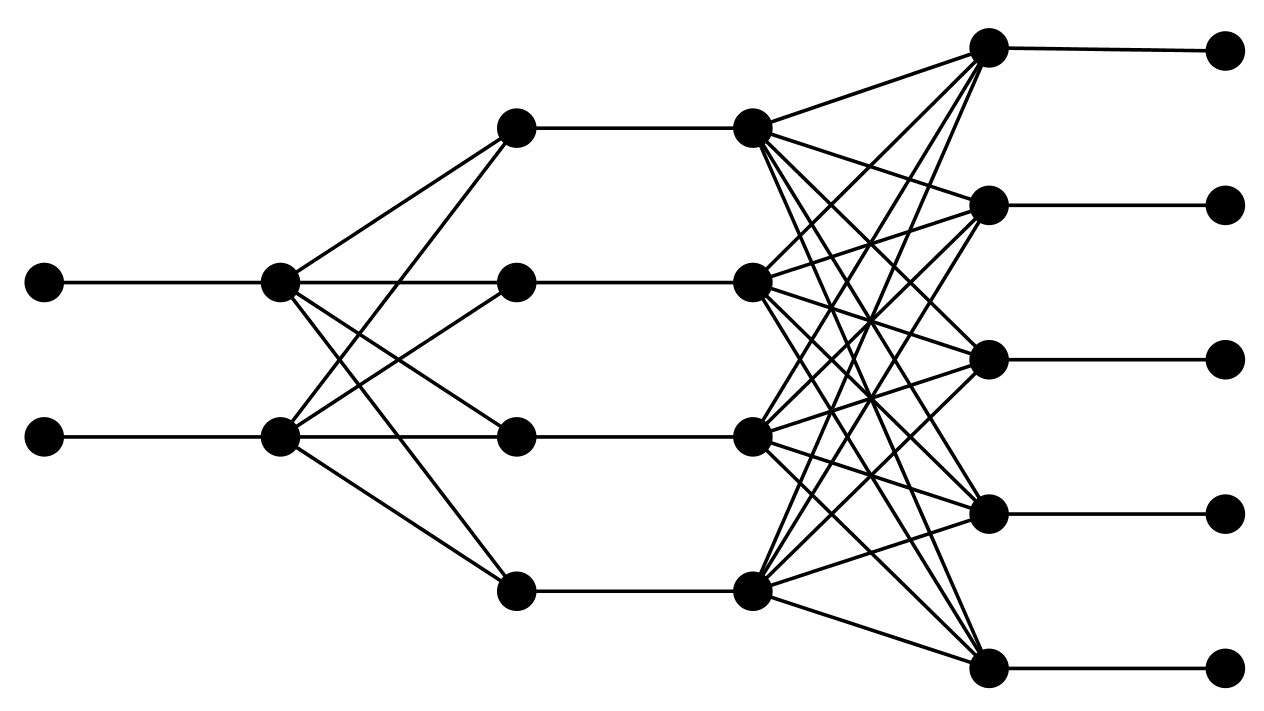} 
\caption{example $\Sb_\bfs$ \label{fig1} }
\end{center}
\end{minipage}
& \begin{minipage}{7.9 cm}
\begin{center}
\includegraphics[width=7.5cm]{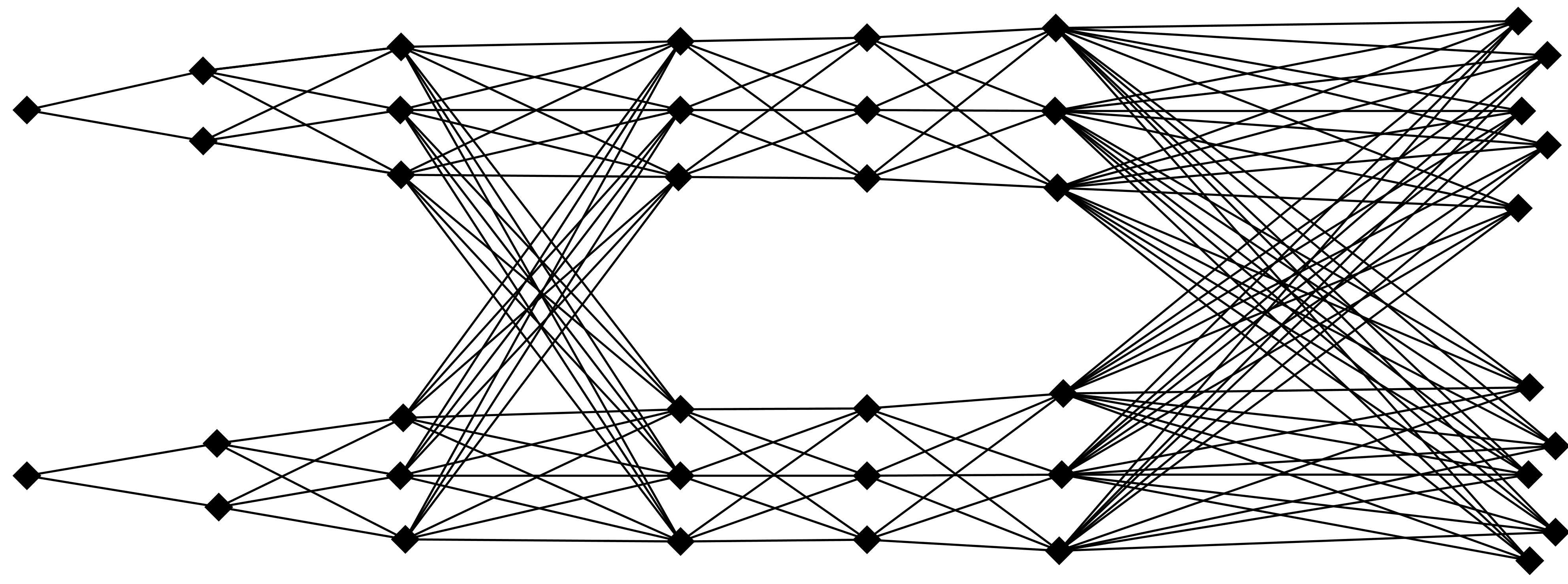} 
\caption{example for $\widehat\Ab_\bfr$,  with $R_{n,1}$ above and $R_{n,2}$ below \label{fig2}}
\end{center}
\end{minipage} 
\end{tabular} 
\end{center}
\end{figure}

\subsection{Anderson type models}

The Anderson type model on these graph structures is then an operator given by the sum of the adjacency operator and a random independent identically distributed potential.
The distribution of a single potential will be called $\nu$.

\begin{assumption}
The single site distribution $\nu$ is a probability distribution on $\RR$ which is compactly supported.
Moreover, by $[\sigma_-,\sigma_+]$ we denote the convex hull of the support of $\nu$ so that in particular
$$
\supp \nu \,\subset\, [\sigma_-, \sigma_+]\;.
$$
\end{assumption}

Then, to construct the random potential we let $(\Omega,\Aaa,\PP)=([\sigma_-,\sigma_+],\BBb,\nu)^{\otimes \GG}$ be the product space with product measure ($\BBb$ denotes the Borel sigma-algebra), where either $\GG=\Sb_{\bfs}$ in the
stretched antitree example or $\GG=\Ab_\bfr$. Then the random potential $\Vv_\omega$ is a multiplication operator on $\ell^2(\GG)$ given by
\begin{equation}\label{eq-Vv_om}
\Vv_\omega \psi(x)\,=\,\omega_x  \psi(x)\,\qtx{where}  \omega=(\omega_x)_{x\in \GG} \in [\sigma_-,\sigma_+]^\GG\;.
\end{equation}
The Anderson model is then given by the random operator
\begin{equation}\label{eq-Hh_om}
\Hh_\omega\,=\,\Aa\,+\, \Vv_\omega\;.
\end{equation}

Now we will describe the region in which we find absolutely continuous spectrum for $\Hh_\omega$ under a certain growth condition.
For Example $1$ we first define
\begin{equation}
I_-:=\{\lambda\in\RR: |\lambda-x|<1 \:\text{for all}\; x\in\supp \nu\}\,,\;I_+:=\{\lambda\in\RR: |\lambda-x|>1\; \text{for all}\; x\in\supp \nu\}
\end{equation}
and for $\lambda \in I_-\cup I_+$ we can then define
\begin{align}
\alpha^\Sb_\lambda\,&=\, \int \frac{x'-\lambda}{(x'-\lambda)(x-\lambda)-1} d\nu(x)\,d\nu(x')\\
\beta^\Sb_{\lambda}\,&=\, \int \frac{1}{(x'-\lambda)(x-\lambda)-1} d\nu(x) d\nu(x')\;
\end{align}
and let
\begin{equation}
I_{\Sb,\nu}:=\{\lambda\in I_+\cup I_-\,:\, \left|\beta^\Sb_{\lambda}+[1-(\alpha^\Sb_\lambda)^2]/\beta^\Sb_\lambda\right|<2 \}\;.
\end{equation}

For Example 2 let us first define the
following unit vectors $\Upsilon,\,\Phi\in\CC^k$ by
\begin{equation}
\Upsilon_j=\begin{cases} 1/\sqrt{k_1} & \qtx{if} j\leq k_1 \\
0 & \qtx{if} k_1 < j \leq k
\end{cases}
\;, \quad
\Phi_j=\begin{cases} 0 & \qtx{if} j\leq k_1+k_2 \\
1/\sqrt{k_3} & \qtx{if} k_1+k_2 < j \leq k
\end{cases}\,.
\end{equation}
Moreover, we let $\DD$ be the set of diagonal $k\times k$ matrices with  $\sigma_-\leq D \leq \sigma_+$
and take $I_0$ to be the set of $\lambda\not\in[\sigma_-,\sigma_+]$ where $D+O\bfa O^*$ is invertible and
$\Upsilon^*(D+O\bfa O^*)^{-1}\Phi\neq 0$ for all $D\in\DD$, i.e.
$$
I_0\,=\,\{\lambda\in \RR\,:\,\lambda\not\in[\sigma_-,\sigma_+]\;\text{and}\; \Upsilon^*(D-\lambda+O\bfa O^*)^{-1}\Phi\neq 0,\;\forall D \in \DD \; \}\;.
$$
Now, for $\lambda\not\in[\sigma_-,\sigma_+]$ let us define the harmonic average
\begin{equation}
h_\lambda\,=\,\left(\int (x-\lambda)^{-1}\,d\nu(x) \right)^{-1}\;, 
\label{eq-def-h}
\end{equation}
and for $\lambda\in I_0$ let
\begin{equation}
\pmat{\alpha_{\lambda}^\Ab & \beta_{\lambda}^\Ab \\ \beta_{\lambda}^\Ab & \delta_\lambda^\Ab}\,=\,
\pmat{\Upsilon^* \\ \Phi^*} (h_\lambda \one+O\bfa O^*)^{-1} \pmat{\Upsilon & \Phi} \;.
\end{equation}
Finally, define
\begin{equation}
I_{\Ab,\nu}\,=\,\{\lambda\in I_0\,:\, \left|\beta^\Ab_\lambda+[1-\alpha^\Ab_\lambda \delta^\Ab_\lambda]/\beta^\Ab_\lambda \right|<2\}\,.
\end{equation}

\begin{Theo}\label{theo-Anderson}
Let $\Hh_\omega$ be the Anderson model on $\ell^2(\GG)$ fulfilling the assumptions above with either $\GG=\Sb_\bfs$ or $\GG=\Ab_\bfr$
as described in Examples 1 and 2 as above.
\begin{enumerate}[{\rm (i)}]
\item For Example 1, assume that $\sum_{n=1}^\infty s_n^{-1}<\infty$, then $I_{\Sb,\nu}\subset \spec(\Hh_\omega)$ almost surely and the spectrum of $\Hh_\omega$ in $I_{\Sb,\nu}$ is almost surely purely absolutely continuous. 
\item For Example 2, in case $r_{3n-1}=0$ for all $n$, assume that $\sum_{n=1}^\infty (\min(r_{3n-2}, r_{3n}))^{-1}<\infty$, and for the case $r_n>0$ for all $n$ let $\sum_{n=1}^\infty (\min(r_{3n-2}, r_{3n-1},r_{3n}))^{-1}<\infty$.\\
In these cases,
$I_{\Ab,\nu}\subset \spec(\Hh_\omega)$ almost surely and the spectrum of $\Hh_\omega$ is almost surely purely absolutely continuous in $I_{\Ab,\nu}$  

\end{enumerate}

\end{Theo}

\begin{rem*}
\begin{enumerate}[{\rm (i)}]
\item In all cases we would like to note that the conditions for obtaining absolutely continuous spectra are fulfilled if $r_n\sim n^d$ (i.e. $s_n\sim n^d$ in Example 1) and $d>2$ in which case the underlying graphs have a $d$-dimensional growth rate for dimension $d>2$.
\item Let us note that for small $\sigma_\pm$ the set $I_{\Sb,\nu}$ is an order $\sigma=\max\{|\sigma_-|,|\sigma_+|\}$ perturbation\footnote{an order $\sigma$ perturbation of a union of $m$ intervals is a union of $m$ similar intervals where the boundary values of each interval change of order $\sigma$.} of 
$I=(-2,-1) \cup (-1,1) \cup (1,2)$ and not empty for small $\sigma$. 
\item For the example $\Ab_\bfr=\widehat \Ab_\bfr$ as above we find  
 $$I_{\Ab,\nu}=\{\lambda\,:\,|h_\lambda|<2,\,|h_\lambda|\neq 1\} \cap \big\{\lambda \in \RR\,:\, \forall \varepsilon\in[\sigma_-,\sigma_+]\,:\,\varepsilon-\lambda \not\in \{-\sqrt{2},0,\sqrt{2} \}\,\big\}$$  
 which for small $\sigma_{\pm}$ is an order $\sigma$ perturbation of
 $$\hat I =(-2,-\sqrt{2})\cup(-\sqrt{2},-1)\cup(-1,0)\cup(0,1)\cup(1,\sqrt{2})
 \cup \sqrt{2}, 2)\;.$$
\end{enumerate}
\end{rem*}

\section{Proof of Proposition~\ref{prop-def-T}}

Here we prove Proposition~\ref{prop-def-T}. Recall that $\WW_n$ and $\tilde \WW_n$ are the cyclic spaces with respect to $V_n$ generated by $\Ups_n$ and $\Phi_n$, respectively. Moreover, $\VV_n=\WW_n+\tilde \WW_n$.

As each eigenvalue in $\WW_n,\,\tilde \WW_n$ has maximum multiplicity $1$, $\lambda\in\spec(V_n|\WW_n\cap \tilde \WW_n)$ means that
$\WW_n$ and $\tilde \WW_n$ contain the same normalized eigenvector $\zeta$ and the multiplicity of $\lambda$ in $\VV_n$ is also just one.
Moreover, let $a:=\zeta^*\Ups_n\neq 0$ and $b:=\zeta^*\Phi_n \neq 0$.
We have to show that the formal appearing order $(\lambda-z)^{-1}$ terms in the definition of $T_{z,n}$ cancel in the limit $z\to\lambda$ and $T_{\lambda,n}$ is well-defined.
With meromorphic functions $\hat \alpha_z,\,\hat \beta_z\,,\hat \gamma_z$ and $\hat \delta_z$ all being equal to $1$ at $z=\lambda$ we find
$$
\alpha_{z,n}=\frac{|a|^2}{\lambda-z} \hat \alpha_z\;,\; \beta_{z,n}=\frac{\bar a b}{\lambda-z} \hat \beta_z\;,\;
\gamma_{z,n}=\frac{\bar b a}{\lambda-z} \hat \gamma_{z}\;,\;\delta_{z,n}=\frac{|b|^2}{\lambda-z} \hat \delta_z\;,\;
$$
which leads to
$$
T_{z,n}\,=\,\pmat{a_n^{-1} (\lambda-z) \hat \beta_z^{-1} \frac{1}{\bar a b} & -a_n\,\frac{a}{b}\, \hat \alpha_z \hat \beta_z^{-1} \\ a_n^{-1} \,\frac{\bar b}{\bar a} \hat \delta_z \hat \beta_z^{-1} & a_n\,\frac{a\bar b}{\lambda-z} \left[\hat \gamma_z -  \hat \delta_z \hat \beta_z^{-1} \hat \alpha_z \right] }\;.
$$
Noting that $\hat \gamma_z -  \hat \delta_z \hat \beta_z^{-1} \hat \alpha_z =0$ at $z=\lambda$ and as it is meromorphic, we find the limit
$$
T_{\lambda,n}\,=\,\pmat{0 & -a_n\;ab^{-1} \\ a_n^{-1}\;\bar b \bar a^{-1} & a_n\,a\bar b\;\frac{d}{dz}\left[ \hat \delta_\lambda \hat \beta_\lambda^{-1} \hat \alpha_\lambda \,-\,\hat \gamma_\lambda  \right]}
$$
where in the lower right entry we take the derivative of $z\mapsto \left[\hat \delta_z \hat \beta_z^{-1} \hat \alpha_z \,-\,\hat \gamma_z \right]$ at $z=\lambda$.
We also have 
$$
\Ups_{z,n}\,=\,\frac{a}{\lambda-z}\zeta\,+\,\hat \Ups_z\;,\quad
\Phi_{z,n}\,=\,\frac{b}{\lambda-z}\zeta\,+\,\hat \Phi_z\;,\quad
$$
where $\hat \Ups_z$ and $\hat \Phi_z$ are holomorphic at $z=\lambda$.
We therefore obtain
$$
a_n \tilde x_{n-1} \Ups_{z,n}\,=\,a_n \tilde x_{n-1} \left(\frac{a}{\lambda-z} \zeta\,+\,\hat\Ups_{z} \right)
$$
and 
$$
 a_{n+1} x_{n+1} \Phi_{z,n}\,=\,\left( \frac{(\lambda-z)\,x_n }{\hat \beta_z\,\bar a\,b} \,-\,a_n \tilde x_{n-1} \hat \alpha_z \hat \beta_z^{-1} \frac{a}{b}\right)\,\left( \frac{b}{\lambda-z}\zeta\,+\,\hat \Phi_z\right)\;.
$$
Adding these terms leads to
$$
\Psi_{z,n}^{x}\,=\, a_n \tilde x_{n-1}\left(\frac{a-a \hat \alpha_z \hat \beta_z^{-1}}{\lambda-z} \right)\,\zeta \,+\,\hat \Psi_{z,n}^{x}
$$
where $\hat \Psi_{z,n}^{x}$ is holomorphic at $z=\lambda$. As $\hat \alpha_z \hat \beta_z$ is holomorphic and equal to $1$  at $z=\lambda$ we see that the limit $\Psi_{\lambda,n}^{x}$ exists.

In the case where $\lambda$ is a simple eigenvalue of $\WW_n$ but not of $\tilde \WW_n$ we find $b=\zeta^* \Phi_n=0$ causing the problem that $\beta_{z,n} \not \to \infty$ for $z\to \lambda$, but
$\alpha_{z,n}\to \infty$ for $z\to \lambda$, hence  $\alpha_{z,n} \beta_{z,n}^{-1}$ has infinite limit.
Vice versa, if $\lambda\in \spec(\tilde \WW_n)$ but $\lambda\not \in \spec (\WW_n)$ we get a blow up of $\delta_{z,n} \beta_{z,n}^{-1}$

Finally, if $\lambda\in\spec(V_n|\VV_n)$ is a double eigenvalue, then we should replace $a$ and $b$ by vectors ${\bfa}$, ${\bfb}$ being the projections of $\Ups_n$ and $\Phi_n$ onto this two dimensional eigenspace. 
Note that this means ${\bfa} \in \WW_n$ and ${\bfb}\in\tilde \WW_n$ and ${\bfa}$ and ${\bfb}$ are not zero and not co-linear.
If ${\bfa}$ and ${\bfb}$ are orthogonal, then $\beta_{z,n}$ stays bounded while $\alpha_{z,n}, \delta_{z,n}\to \infty$ for $z\to \lambda$ and we have the same problem as before.
If they are not orthogonal, then the products $\bar a b$ and $\bar b a$ are replaced by $\bfa^* {\bfb}$ and ${\bfb}^* {\bfa}$ in the relations above for $\beta_{z,n}$ and $\gamma_{z,n}$ and we obtain
$$
\gamma_{z,n}-\delta_{z,n}\beta_{z,n}^{-1}\alpha_{z,n}\,=\, \frac{|{\bfb}^* {\bfa}|^2 \hat \gamma_z - |{\bfa}|^2 \,|{\bfb}|^2 \;\hat \delta_z \hat \beta_z^{-1} \hat \alpha_z}{(\lambda-z){\bfa}^* {\bfb} }\;.
$$
For $z\to \lambda$ the numerator converges to $\left|{\bfb}^* {\bfa}\right|^2 - |{\bfa}|^2\;|{\bfb}|^2 $ which is negative and not zero by Cauchy-Schwarz as ${\bfa}$ and ${\bfb}$ are not co-linear. Hence the expression converges to $\infty$ for $z \to \lambda$ and the matrix $T_{\lambda,n}$ is not defined.

\section{Eigenfunctions with finite support}

Here we want to classify the eigenfunctions of finite support other than the ones supported in one shell coming from $V_n|\VV_n^\perp$.
Thus, we may assume $\VV_n=\ell^2(S_n)$ and each eigenvector of $\VV_n$ overlaps with $\Phi_n$ or $\Ups_n$.

\begin{lemma}\label{lem-boundary}
 Assume $\Hh \Psi = \lambda \Psi$, $\lambda\in\RR$ and as above, let $x_n=\Ups_n^*\Psi_n$ and $\tilde x_n=\Phi_n^*\Psi_n$.
 Assume that $T_{\lambda,n}$ is not defined, i.e. $\lambda \in A_n \cap \RR$.
 Then, one has the following local 'boundary conditions' at $n$:
 At the left to the shell $S_n$ one finds $\smat{a_n x_n \\ \tilde x_{n-1}}=C \ovr{x}^{(-)}_{\lambda,n}$, i.e.
  $$
  \pmat{a_n x_{n} \\ \tilde x_{n-1}}\,=\,C \cdot \pmat{b^2_n \Ups_n^*(V_n-\lambda)^{-1} \Ups_n \\ 1} \qtx{or} \tilde x_{n-1}=0
  $$
  depending on whether $\Ups_n^*(V_n-\lambda)^{-1} \Ups_n$ exists or not.
  At the right to the shell $S_n$ one finds $\smat{a_{n+1} x_{n+1} \\ \tilde x_{n}}=C \ovr{x}^{(+)}_{\lambda,n}$, i.e.
  $$
   \pmat{a_{n+1} x_{n+1} \\ \tilde x_n}\,=\,C\cdot \pmat{1 \\ \Phi_n^*(V_n-\lambda)^{-1}\Phi_n } \qtx{or} x_{n+1}=0
  $$
  depending on whether $\Phi_n^*(V_n-\lambda)^{-1}\Phi_n$ exists or not.\\
\end{lemma}
  Note, that formally the equations $\tilde x_{n-1}=0$ or $x_{n+1}=0$  can be included in the other equations with $|\Ups_n^*(V_n-\lambda)^{-1} \Ups_n|= \infty$ or
  $|\Phi_n^*(V_n-\lambda)^{-1}\Phi_n|=\infty$, respectively, interpreted as an equation on projective space.
\begin{proof}
 Start with the equation
 \begin{equation} \label{eq-eig}
 0\,=\,(V_n-\lambda)\Psi_n\,=\,a_{n+1}\Phi_n\,x_{n+1}\,+\,a_n \Ups_n\,\tilde x_{n-1}\,
 \end{equation}
 which follows from $\Hh\Psi=\lambda \Psi$.
 There are four cases where $T_{\lambda,n}$ is not defined.\\[.2cm]
 {\bf Case 1}: $\lambda\not\in\spec(V_n|\VV_n), \beta_{\lambda,n}=0=\gamma_{\lambda_n}$, i.e. $\Phi_n^*(V-\lambda)^{-1} \Ups_n=0=\Ups_n^*(V-\lambda)^{-1} \Phi_n$.
 Multiplying \eqref{eq-eig} with $\Phi_n^*(V_n-\lambda)^{-1}$ or $\Ups_n^*(V_n-\lambda)^{-1}$ leads directly to the claimed equations.\\[.2cm]
 {\bf Case 2}: There exist ${\bfa}, {\bfb} \in \ker(V_n-\lambda)$ with ${\bfa}^*\Ups_n\neq 0, {\bfa}^*\Phi_n=0={\bfb}^*\Phi_n, {\bfb}^*\Phi_n\neq 0$. In this case, neither
 $\Ups_n^*(V_n-\lambda)^{-1}\Ups_n$ nor $\Phi_n^*(V_n-\lambda)^{-1}\Phi_n$ exist.
 Multiplying \eqref{eq-eig} with ${\bfa}^*$ from the left, one immediately gets $\tilde u_{n-1}=0$. Multiplying with ${\bfb}^*$ from the left gives $u_{n+1}=0$.\\[.2cm]
 {\bf Case 3}: $V_n-\lambda|\tilde \WW_n$ is invertible, and there exists ${\bfa}\in \ker(V_n-\lambda)$ with ${\bfa}^*\Ups_n\neq 0, {\bfa}^*\Phi_n=0$. In this case $\Phi_n^*(V_n-\lambda)^{-1} \Phi_n$ exists (calculated by restriction to $\tilde \WW_n$).
 Multiplying \eqref{eq-eig} by ${\bfa}^*$ from the left first gives $\tilde u_{n-1}=0$. Plugging this back in \eqref{eq-eig} one obtains an equation within the vector space $\tilde \WW_n$. Multiplying by $\Phi_n^*(V_n-\lambda)^{-1}$ from the left then
 gives the corresponding condition at the right of $S_n$.\\[.2cm]
 {\bf Case 4}: $V_n-\lambda|\WW_n$ is invertible, and there exists ${\bfb}\in \ker(V_n-\lambda)$ with ${\bfb}^*\Phi_n\neq 0, {\bfb}^*\Ups_n=0$. Here, $\Ups_n^*(V_n-\lambda)^{-1} \Ups_n$ exists.
 This case is symmetric to case 3 and everything follows analogously.
\end{proof}

Now we have the ingredients to classify the compactly supported eigenvectors of $\Hh|\VV$ or $\Hh^+|\VV^+$.
\begin{proof}[Proof of Theorem~\ref{th-compact-eig}]
 In part (i) we have $\lambda\in \RR\cap A_{l} \cap A_{m+1}$ and $\lambda \not \in B_{l+1,m}$. By the assumptions, there exists a (unique) solution $(x_{n+1},\tilde x_n)_{n=l}^m$ of the transfer matrix equation with
 $\smat{a_{l+1} x_{l+1} \\ x_l}=\ovr{x}^{(+)}_{\lambda,l}$ and $\smat{a_{m+1} x_{m+1} \\ x_m}=C \ovr{x}^{(-)}_{\lambda,m}$.
 For $l<m\leq n$ let $\Psi_n:=a_{n+1} x_{n+1} \Phi_{\lambda,n}+a_n \tilde x_{n-1} \Ups_{\lambda} \lambda$, cf. \eqref{eq-def-psi_z}.
 Moreover, using the case classification in the proof of Lemma~\ref{lem-boundary} we let 
$$
 \Psi_l:=a_{l+1} u_{l+1} (V_l-\lambda)^{-1} \Phi_l \in \tilde \WW_l \quad \text{in Case 1 and 3}
$$
(which are the cases where $V_l-\lambda|\tilde \WW_l$ is invertible) and
$$
 \Psi_l \in \tilde \WW_l \qtx{such that} V_n\Psi_l = \lambda \Psi_l \qtx{and} \Phi_l^* \Psi_l=u_l \quad \text{in Case 2 and 4.}
$$
Note that also in the Cases 2 and 4, $\Psi_l$ is uniquely defined and a multiple of the vector ${\bfb}$ as in the proof of Lemma~\ref{lem-boundary}.
Similarly, we let
$$
\Psi_{m+1}:= a_m \tilde u_{m}\,(V_{m+1}-\lambda)^{-1} \Ups_m \quad \text{in Case 1 and 4}
$$
(which are the cases where $V_m-\lambda|\WW_{m+1}$ is invertible) and
$$
\Psi_{m+1} \in \WW_{m+1} \qtx{such that} V_{m+1} \Psi_{m+1}=\lambda \Psi_{m+1} \qtx{and} \Ups_{m+1}^*\psi_{m+1}=u_{m+1}
$$ in Case 2 and 3.
Moreover, let $\Psi_n=\nul$ for $n<l$ or $n>m+1$. Then, it is easy to check that $\Hh\Psi=\lambda\Psi$ for $\Psi=\bigoplus_n \Psi_n$.
Any other eigenfunction $\Theta\in\VV$ for the eigenvalue $\lambda$, which is supported on $\bigsqcup_{n=l}^{m+1} S_n$ satisfies $\Ups_{n+1}^* \Theta_{n+1}=C u_{n+1}$ and $\Phi_n^* \Theta_n=C \tilde u_n$ for $l\leq n\leq m$ and some fixed $C$ by
Lemma~\ref{lem-boundary}. Let $\psi=C\Psi-\Theta$, from \eqref{eq-eig-1} it follows that $\psi_n=0$ except for $n=l$ and $n=m+1$, and one has $\Phi_l^* \psi_l=0=\Ups_{m+1}^*\psi_{m+1}=0$. 
This leads to $V_l \psi_l=\lambda \psi_l$ and $V_{m+1} \psi_{m+1}=\lambda\psi_{m+1}$.
The equation at $l-1$ or $m+2$ also reveals $\Ups_l^* \psi_l=0=\Phi_{m+1}^* \psi_{m+1}$. All together we find $\psi_l \in \VV_l^\perp$ and $\psi_{m+1}\in \VV_{m+1}^\perp$. As $\psi\in\VV$ this gives $\psi=\nul$ and thus $\Theta=C\Psi$ which finishes the proof of part (i).
Part (ii) works with the same construction.

For part (iii), let us focus again on the operator $\Hh$, the calculations for $\Hh^+$ are analogue.
Assume $\Hh \Theta = \lambda \Theta$, $\Theta\in\VV$, $\Theta_l \neq \nul, \Theta_m \neq \nul$, and $\Theta_n=\nul$ for $n<l$ and $n>m$. Moreover, let $u_n=\Ups_n^*\Theta_n $ and $\tilde u_n=\Phi_n^*\Theta_n$.
Then, the eigenvalue equation \eqref{eq-eig-1} at $n=l-1$ gives $u_l=0$ and of course we have $\tilde u_{l-1}=0$. Similarly, we find $\tilde u_m=0=u_{m+1}$.

\noindent {\bf Claim: Neither $T_{\lambda,l}$ nor $T_{\lambda,m}$ are defined.}\\
If $T_{\lambda,l}$ were well defined, then $\smat{a_{l+1} u_{l+1}\\ \tilde u_l}=T_{\lambda,l} \nul = \nul$ and by similar considerations as above we get $V_l\Theta_l=\lambda \Theta_l$ and $\Theta_l\in \VV_l^\perp$, hence $\Theta_l=\nul$, which contradicts our assumptions.
Thus $T_{\lambda,l}$ is not defined. A similar argument gives that $T_{\lambda,m}$ is not defined.

Now, if $T_{\lambda,n}$ is well defined for all $l<n<m$, then part (i) shows $\Theta=C\Psi$ for some $\Psi$ as constructed above. Otherwise, let $l < m' < m$ be the smallest number such that $T_{\lambda,m'}$ is not defined.
Lemma~\ref{lem-boundary} and similar arguments as above give some eigenfunction $\Psi$ supported on $\bigcup_{n=l}^{m'} S_n$ as in (i) and a constant $C$ such that $\Theta- C\Psi$ is supported on $\bigcup_{n=m'}^m S_n$.
Iterating the process shows that $\Theta$ is a linear combination of vectors $\Psi$ as constructed in part (i).

Parts (iv) and (v) are very analogue to (i).
\end{proof}


\section{Basic Green's function identities \label{sec:basic-green}}

\subsection{Finite cut off operators}

Let $\Hh_{N,c}$ be the restriction of $\Hh^+$ to $\bigoplus_{k=1}^n \ell^2(S_k)=\bigoplus_{k=1}^n \CC^{s_k}$ with boundary conditions $\Psi_0=\nul$ and $a_{N+1} \Ups_{N+1}^*\Psi_{N+1}=c\Phi_N^*\Psi_N$.
Then, $\Hh_{N,0}$ has Dirichlet boundary conditions and 
$$
\Hh_{N,c}=\Hh_{N,0}-c \,P_N \Phi_N \Phi_N^* P_N^*
$$ 
is a self-adjoint matrix for real $c$.
For $c\in \RR,\; z\not\in\spec(\Hh_{n,c})$ we define
\begin{equation}
m_{N,c}(z)\,:=\,\langle P_1 \Ups_1;(\Hh_{N,c}-z)^{-1} P_1\Ups_1\rangle\,=\,\Ups_1^* P_1^* (\Hh_{N,c}-z)^{-1} P_1 \Ups_1 \label{eq-def-m}
\end{equation}
and 
\begin{equation}
\tilde m_{N,c}(z)\,:=\, \langle P_1 \Phi_1; (\Hh_{N,c}-z)^{-1} P_1 \Phi_1\rangle\,=\,\Phi_1^* P_1^* (\Hh_{N,c}-z)^{-1} P_1 \Phi_1 \label{eq-def-tm}
\end{equation}

Furthermore, let $(x_{z}^{(N,c)},\tilde x_{z}^{(N,c)})$ be the solution at the spectral parameter $z$ with the right boundary condition $$a_{N+1}\,x_{z,N+1}^{(N,c)}=c\,,\quad\tilde x^{(N,c)}_{z,N}=1\;.$$
One should note that formally changing $a_1$ does not change the operator $\Hh_{N,c}$, it does also not change $u_{z,n},\,\tilde u_{z,n},\, w_{z,n},\,\tilde w_{z,n},\,x_{z,n},\,\tilde x_{z,n}$ for positive $n\geq 1$.
However, $\tilde x_{z,0}$ is proportional to $1/a_1$, thus the products $a_1 \tilde x_{z,0}$ and $\det(T_{z,0,n})$ are independent of the choice of $a_1$.
\begin{lemma}\label{lem-bar-u}
 We have
 $$
\frac{u_{z,n+1}}{\overline{u}_{\bar z,n+1}} \,=\,
\frac{w_{z,n+1}}{\overline{w}_{\bar z,n+1}} \,=\,
\frac{\tilde u_{z,n}}{\overline{\tilde u}_{\bar z,n}}\,=\, 
\frac{\tilde w_{z,n}}{\overline{\tilde w}_{\bar z,n}}\,=\, \det(T_{z,0,n})\;,
$$ 
$$
a_{n+1}\,\left(\overline{u}_{\bar z,n+1} \,\tilde w_{z,n}\,-\,w_{z,n+1}\,\overline{\tilde u}_{\bar z,n} \right)\,=\,1
$$
and 
$$
\overline{x}^{(N,c)}_{\bar z,n+1} \,=\,x^{(N,c)}_{z,n+1} \det(T_{z,n,N})\;,\qquad
\overline{\tilde x}^{(N,c)}_{\bar z,n}\,=\, \tilde x^{(N,c)}_{z,n} \det(T_{z,n,N})\;.
$$ 
\end{lemma}
\begin{proof}
Note that  $\bar T_{\bar z,n} = T_{z,n} \beta_{z,n} \gamma^{-1}_{z,n}=T_{z,n} / \det(T_{z,n})$ 
and $u_{z,1},\tilde u_{z,0}, w_{z,1},\tilde w_{z,0}, x^{(N,c)}_{z,N+1}$ and $\tilde x^{(N,c)}_{z,N}$ are all real, e.g. in $\RR$.
With \eqref{eq-T-u-v} this immediately gives the identities above. 
\end{proof}

For complex $z$ and real $c$ the solutions $(x^{(N,c)}_z,\tilde x^{(N,c)}_z)$ and $(u_z,\tilde u_z)$ can not be co-linear, otherwise $z$ would be an eigenvalue of $\Hh_{N,c}$. Hence, $\tilde x^{(N,c)}_{z,0}\neq 0$.
Note that for an operator $\Gg$ on $\ell^2(\GG)$ or $\bigoplus_{k\in I} \ell^2(S_k)$ with $m,n\in I$ the operator $P_m^* \Gg P_n$ makes sense as a $s_m \times s_n$ matrix.
Our aim is now to get expressions for the resolvent, in particular $P_m^* (\Hh_{N,c}-z)^{-1} P_n$.

We have to solve $(\Hh_{N,c}-z)\Psi = P_n \varphi$, where $1\leq n \leq N,\,\varphi\in \ell^2(S_n)$. Furthermore, for the unique solution $(\Psi_1,\ldots,\Psi_N)$ we let $x_n:=\Ups_n^*\Psi_n$ and $\tilde x_n=\Phi_n^*\Psi_n$. 
For simpler notation we will use $x_{z,n}$ and $\tilde x_{z,n}$ instead of $x^{(N,c)}_{z,n}$ and $\tilde x^{(N,c)}_{z,n}$.

Using the right and left boundary conditions and the fact that $((\Hh-z)\Psi)_m=0$ for $m\neq n$, we obtain:
\begin{equation}\label{eq-basic-psi}
 \Psi_m\,=\,\Phi_{z,m}\,a_{m+1} x_{m+1} + \Ups_{z,m} a_m \tilde x_{m-1} \qtx{for} m\neq n
\end{equation}
and 
\begin{equation}\label{eq-rel-x-u-w}
 x_m = \begin{cases}
        c_1 u_{z,m} & \;\;\text{for}\;\; m \leq n \\
        c_2 x_{z,m} & \;\;\text{for}\;\; m > n 
       \end{cases}
\qtx{and} \tilde x_m = \begin{cases}
                    c_1 \tilde u_{z,m} & \;\;\text{for}\;\; m < n \\
		    c_2 \tilde x_{z,m} & \;\;\text{for}\;\; m \geq n 
                   \end{cases}
\end{equation}
for some ($z$ and $\varphi_n$-dependent) constants $c_1, c_2$.
The identity $(\Hh_{N,c}-z)\Psi)_n=\varphi$ gives
$$
\Psi_n\,=\,\Phi_{z,n}\,a_{n+1} x_{n+1} + \Ups_{z,n} a_n \tilde x_{n-1}+(V_n-z)^{-1} \varphi
$$
leading to
$$
x_n\,=\,\beta_{z,n} a_{n+1} x_{n+1}+\alpha_{z,n} a_n \tilde x_{n-1} + \Ups_{\bar z,n}^* \varphi
$$
and 
$$
\tilde x_n\,=\,\delta_{z,n} a_{n+1} x_{n+1} + \gamma_{z,n} a_n \tilde x_{n-1}+\Phi_{\bar z,n}^* \varphi\,.
$$
Note that $\Ups_{\bar z,n}^*\varphi_n = \Ups_n^*[(V_n-\bar z)^*]^{-1} \varphi_n=\Ups_n^*(V_n-z)^{-1}\varphi$.
Similar calculations as above lead to
$$
\pmat{a_{n+1} x_{n+1} \\ \tilde x_n}\,=\,
T_{z,n} \pmat{a_n x_n \\ \tilde x_{n-1}}\,-\, \pmat{ \beta_{z,n}^{-1}\Ups_{\bar z,n}^* \varphi \\ (\beta_{z,n}^{-1} \delta_{z,n} \Ups_{\bar z,n}^*-\Phi_{\bar z,n}^*)\varphi}
$$
implying  that
$$
c_2 \pmat{a_{n+1} x_{z,n+1} \\ \tilde x_{z,n}}\,=\,c_1 \pmat{a_{n+1} u_{z,n+1} \\ \tilde u_{z,n}}
\,-\,\pmat{ \beta_{z,n}^{-1} \Ups_{\bar z,n}^* \varphi \\ (\beta_{z,n}^{-1} \delta_{z,n} \Ups_{\bar z,n}^*-\Phi_{\bar z,n}^*)\varphi}\;.
$$
For $\im(z)>0$ and $c\in \RR$ the vectors $(a_{n+1} u_{z,n+1}\,, \tilde u_{z,n})$ and $(a_{n+1} x_{z,n+1}, \,\tilde x_{z,n})$
can not be co-linear, otherwise \eqref{eq-eig-1} would give an eigenvector $(\Psi_1,\ldots,\Psi_N)$ of $\Hh_N$ with eigenvalue $z$.
Therefore, one can uniquely solve for $c_1,\, c_2$ and we find
\begin{align}
c_1\,&=\,\frac{\tilde x_{z,n} \beta_{z,n}^{-1} \Ups_{\bar z,n}^* \varphi - a_{n+1} x_{z,n+1} (\beta_{z,n}^{-1} \delta_{z,n} \Ups_{\bar z,n}^*-\Phi_{\bar z,n}^*)\varphi  }{a_{n+1} (u_{z,n+1} \tilde x_{z,n} - x_{z,n+1} \tilde u_{z,n})} \label{eq-c1-1} \\
c_2\,&=\,\frac{\tilde u_{z,n} \beta_{z,n}^{-1}\Ups_{\bar z,n}^* \varphi - a_{n+1} u_{z,n+1} (\beta_{z,n}^{-1} \delta_{z,n} \Ups_{\bar z,n}^*-\Phi_{\bar z,n}^*)\varphi  }{a_{n+1} (u_{z,n+1} \tilde x_{z,n} - x_{z,n+1} \tilde u_{z,n})}\;. \label{eq-c2-1}
\end{align}
Note that the denominator equals
$$
 \det\pmat{a_{n+1} u_{z,n+1} & a_{n+1} x_{z,n+1}  \\ \tilde u_{z,n} & \tilde x_{z,n} }\,=\, \det\left(T_{z,0,n}\pmat{a_1& a_1 x_{z,1} \\ 0 & \tilde x_{0,z}}\right)\;.
$$
As $(x_z,\tilde x_z)$ solve the transfer matrix equation, we find $\tilde x_{z,n}-a_{n+1} x_{z,n+1} \delta_{z,n}=\gamma_{z,n} a_n \tilde x_{z,n-1}$ and a similar equation replacing $(x_z,\tilde x_z)$ by $(u_z,\tilde u_z)$. This leads to
\begin{align}
c_1\,&=\,\frac{(\beta_{z,n}^{-1}\gamma_{z,n} a_n \tilde x_{z,n-1} \Ups_{\bar z,n}^*\,+\,a_{n+1} x_{z,n+1} \Phi_{\bar z,n}^*)\,  \varphi  }{\det(T_{z,0,n})\, a_1 \tilde x_{z,0}} 
\,=\,\frac{(\Psi^{x}_{\bar z,n})^*\varphi}{a_1\,\overline{\tilde x}_{\bar z,0}}\label{eq-c1-2} \\
c_2\,&=\,\frac{(\beta_{z,n}^{-1}\gamma_{z,n} a_n \tilde u_{z,n-1} \Ups_{\bar z,n}^*\,+\,a_{n+1} u_{z,n+1} \Phi_{\bar z,n}^*)\,  \varphi  }{\det(T_{z,0,n})\, a_1 \tilde x_{z,0}}\,=\,\frac{(\Psi^{u}_{\bar z,n})^*\varphi}{a_1\,\tilde x_{z,0}} \label{eq-c2-2}
\end{align}
for the last equations we used Lemma~\ref{lem-bar-u} and \eqref{eq-def-psi_z}.
We have two special cases,
$$
c_1\,=\,\frac{1}{\det(T_{z,0,n})a_1\tilde x_{z,0}}\;\cdot\;\begin{cases}
                                          \beta_{z,n}^{-1} \gamma_{z,n} x_{z,n}\, & \;\;\text{if}\;\;\varphi_n=\Ups_n \\
                                          \tilde x_{z,n} & \;\;\text{if}\;\; \varphi_n= \Phi_n
                                        \end{cases}
$$
and a similar equation holds when replacing $c_1$ by $c_2$ and $(w,\tilde w)$ by $(u,\tilde u)$. These identities lead to 
the following important equations:
\begin{align}
& \langle P_m \Ups_m\,;\, (\Hh_{N,c}-z)^{-1} P_n \Ups_n\rangle\,=\,\frac{\beta_{z,n}^{-1}\gamma_{z,n} }{\det(T_{z,0,n}) a_1 \tilde x^{(N,c)}_{z,0}} \;\cdot\; \begin{cases}
                                                                                                                                                  u_{z,m} x^{(N,c)}_{z,n} & \;\;\text{if}\;\; m\leq n \\
                                                                                                                                                  x^{(N,c)}_{z,m} u_{z,n} & \;\;\text{if}\;\; m\geq n
                                                                                                                                                 \end{cases} 
                                                                                                                                                 \label{eq-mn}\\
& \langle P_m \Ups_m\,;\, (\Hh_{N,c}-z)^{-1} P_n \Phi_n\rangle\,=\,\frac{1}{\det(T_{z,0,n}) a_1 \tilde x^{(N,c)}_{z,0}} \;\cdot\; \begin{cases}
                                                                                                                                                  u_{z,m} \tilde x^{(N,c)}_{z,n} & \;\;\text{if}\;\; m\leq n \\
                                                                                                                                                  x^{(N,c)}_{z,m} \tilde u_{z,n} & \;\;\text{if}\;\; m> n
                                                                                                                                                 \end{cases} 
                                                                                                                                                 \label{eq-mtn}\\
& \langle P_m \Phi_m\,;\, (\Hh_{N,c}-z)^{-1} P_n \Ups_n\rangle\,=\,\frac{\beta_{z,n}^{-1}\gamma_{z,n} }{\det(T_{z,0,n}) a_1 \tilde x^{(N,c)}_{z,0}} \;\cdot\; \begin{cases}
                                                                                                                                                  \tilde u_{z,m} x^{(N,c)}_{z,n} & \;\;\text{if}\;\; m< n \\
                                                                                                                                                  \tilde x^{(N,c)}_{z,m} u_{z,n} & \;\;\text{if}\;\; m\geq n
                                                                                                                                                 \end{cases} \label{eq-tmn} \\
& \langle P_m \Phi_m\,;\, (\Hh_{N,c}-z)^{-1} P_n \Phi_n\rangle\,=\,\frac{1 }{\det(T_{z,0,n}) a_1\tilde x^{(N,c)}_{z,0}} \;\cdot\; \begin{cases}
                                                                                                                                                  \tilde u_{z,m} \tilde x^{(N,c)}_{z,n} & \;\;\text{if}\;\; m\leq n \\
                                                                                                                                                  \tilde x^{(N,c)}_{z,m} \tilde u_{z,n} & \;\;\text{if}\;\; m\geq n
                                                                                                                                                 \end{cases} \label{eq-tmtn}
\end{align}

In particular, we find the following expressions
\begin{equation}\label{eq-def-m-fct}
 m_{N,c}(z)\,=\,\frac{x^{(N,c)}_{z,1}}{a_1\tilde x^{(N,c)}_{z,0}}\;,\quad
 \tilde m_{N,c}(z)\,=\,\frac{\delta_{z,1} \tilde x^{(N,c)}_{z,1}}{\gamma_{z,1} a_1\,\tilde x^{(N,c)}_{z,0}}\,=\,
 \frac{\tilde u_{z,1}\tilde x^{(N,c)}_{z,1} }{\beta^{-1}_{z,1}\gamma_{z,1} a_1\tilde x^{(N,c)}_{z,0}}.
\end{equation}
The first equation leads to
\begin{equation}\label{eq-rel-m-w}
 \frac{x^{(N,c)}_{z,n}}{a_1\tilde x^{(N,c)}_{z,0}}\,=\,m_{N,c}(z) u_{z,n}+w_{z,n} \;,\qquad
 \frac{\tilde x^{(N,c)}_{z,n}}{a_1\tilde x^{(N,c)}_{z,0}}\,=\,m_{N,c}\tilde u_{z,n}+\tilde w_{z,n}\;.
\end{equation}
One only needs to check the first equation for the case $n=1$ and the second one for $n=0$. The general case follows then by linearity from the transfer matrix equation.
In particular, we obtain
\begin{equation}\label{eq-def-psi-na}
  \Psi^{x^{(N,c)}}_{z,n}\,/\,(a_1\tilde x^{(N,c)}_{z,0})\,=\,m_{N,c}(z)\,\Psi^u_{z,n}\,+\,\Psi^w_{z,n}\,=:\,\Psi^{(N,c)}_{z,n}\;.
\end{equation}
where the latter equation is the definition of $\Psi^{(N,c)}_{z,n}$.

 
Using the vector solutions $\Psi^u_{z,n}$ and $\Psi^{(N,c)}_{z,n}$, we find some simple expressions for the part of the resolvents:
\begin{prop}\label{prop:green}
Let $m<n\leq N$. Then,
\begin{align}
 & P_m^* (\Hh_{N,c}-z)^{-1} P_n\,=\, \Psi^u_{z,m}\,(\Psi^{(N,c)}_{\bar z,n})^* \label{eq-Green-1} \\
 & P_n^* (\Hh_{N,c}-z)^{-1} P_m\,=\, \Psi^{(N,c)}_{z,n}\,(\Psi^u_{\bar z,n})^* \label{eq-Green-2} 
\end{align}
where the products of the type $\Psi'\Psi^*$ have to be considered as matrix product of a column vector with a row vector giving a matrix. 
In the case $m=n$ we find
\begin{align} 
P_n^*(\Hh_{N,c}-z)^{-1} P_n \,=&\;
m_{N,c}(z)\; \Psi^{u}_{z,n} (\Psi^{u}_{\bar z,n})^* \,+\,(V_n-z)^{-1} \notag \\ & +\, \Phi^w_{z,n}\,(\Psi^u_{\bar z,n})^*\,+\,\Ups^u_{z,n}\,(\Psi^w_{\bar z,n})^*
\label{eq-Green-3}
\end{align}
Particularly, we have
\begin{equation}
 ((\Hh_{N,c}-z)^{-1}\,P_1\Ups_1)_n\,=\,\Psi^{x^{(N,c)}}_{z,n}\,/\,(a_1 \tilde x_{z,0})\,=\,\Psi^{(N,c)}_{z,n}
\end{equation}
for $n=1,\ldots,N$.
\end{prop}

\begin{proof}
For $m<n$ we find $\Psi_m=c_1 \Psi^u_{z,m}$ and for $m>n$ we have $\Psi_m=c_2 \Psi^x_{z,m}$
which together with \eqref{eq-basic-psi}, \eqref{eq-c1-2} and \eqref{eq-c2-2} shows \eqref{eq-Green-1} and \eqref{eq-Green-2}.

The expression for $\Psi_n$ gives
$$
\Psi_n\,=\,(V_n-z)^{-1}\varphi\,+\,\frac{\Phi^x_{z,n} (\Psi^u_{z,n})^*}{a_1\,\overline{\tilde x}_{\bar z,0}}\,+\,
\frac{\Ups^u_{z,n} (\Psi^x_{z,n})^*}{a_1\,\tilde x_{z,0}}
$$
and using \eqref{eq-rel-m-w} we find
$$
\frac{1}{a_1 \tilde x_{0,z}}\Phi^x_{z,n}\,=\,m_{N,c}(z)\,\Phi^u_{z,n}\,+\,\Phi^w_{z,n}
$$
and the analogue equation holds for $\Psi^x_{z,n}$. Noting that $\bar m_{N,c}(\bar z) = m_{N,c}(z)$ we finally obtain
\begin{align*}
 & \Psi_n\,=\,\left((V_n-z)^{-1} \,+\, m_{N,c}(z)\; \Psi^{u}_{z,n} (\Psi^{u}_{\bar z,n})^* \,+\,\Phi^w_{z,n}\,(\Psi^u_{\bar z,n})^*\,+\,\Ups^u_{z,n}\,(\Psi^w_{\bar z,n})^*\right)\,\varphi
\end{align*}
This proves \eqref{eq-Green-3}.

In the case of $n=1$ and $\varphi=\Ups_1$ we find $(\Psi^u_{\bar z,1})^* \Ups_1=\bar u_{z,1}=1$ and 
$(\Psi^w_{\bar z,1})^*\Ups_1=\bar w_{z,1}=0$ giving
$\Psi_m=\Psi^x_{z,m}\,/\,(a_1\tilde x_{z,0})=\Psi^{(N,c)}_{z,m}$ for all $m\geq 1$.
\end{proof}

\subsection{Unique self-adjointness and limit point\label{sec:selfadj}}

The main criterion for essential self-adjointness or existence of self-adjoint extensions are the so called deficiency indices.
So let us consider the operators $\Hh_{\min}^\flat$, $\flat$ symbolizing $+$, $-$ or no index at all.
The deficiency indices are the dimensions of the kernels of $(\Hh_{\min}^\flat)^*\pm i=\Hh^\flat_{\max}\pm i$, i.e. let us define
$$
d^\flat_{+}=\dim\; \ker(\Hh^\flat_{\max}-i) \qtx{and} d^\flat_-=\dim\; \ker (\Hh^\flat_{\max}- (-i)\,).
$$
By general theory, the dimension of the kernels $\Hh^\flat_{\max} - z$ are constant for $z$ varying in the upper and lower half planes $\CC^{\pm}=\{z\in\CC\,:\,\pm \im(z) > 0\}$. We also have $d_\pm =d^+_\pm + d^-_\pm$.
Moreover, $\Hh^\flat_{\max}$ is the unique self-adjoint closure of $\Hh^\flat_{\min}$ if and only if both deficiency indices are zero.

\vspace{.2cm}

We start with the following identity.
\begin{lemma}\label{lem-sum-est}
Let $(w,\tilde w)$ be any solution of the transfer matrix equation at $z$ and let $\Psi^x_z$ the corresponding solution vector as in \eqref{eq-def-psi_z}.
Then one finds
$$
\im(z)\;\sum_{k=m}^n \|\Psi^x_{z,n}\|^2\,=\,\im\left(a_{n+1}\, \overline x_{n+1}\,\tilde x_n\,-\,a_m\, \overline x_m\,\tilde x_{m-1} \right)\;.
$$
\end{lemma}

\begin{proof}
Using $\Phi_n^*\Psi^x_{z,n}=\tilde x_n$ and $\Ups_n^* \Psi^x_{z,n}=x_n$ one has
$$
z\,\Psi_{z,n}^x\,=\,(\Hh \Psi_z^x)_n\,=\,-a_{n+1} x_{n+1} \Phi_n\,-\,a_n \tilde x_{n-1} \Ups_n\,+\,V_n \Psi_{z,n}\;.
$$
Multiplying with $(\Psi^x_{z,n})^*$ from the left gives
$$
z\,\| \Psi^x_{z,n} \|^2\,=\, -a_{n+1} x_{n+1}\,\overline{\tilde x}_n\,-\,a_n\,\overline{x}_n\,\tilde x_{n-1}\,+\,(\Psi^x_{z,n})^* \,V_n\,\Psi^x_{z,n}\;.
$$
Taking imaginary parts and summing up proves the lemma.
\end{proof}
The following proposition corresponds to Proposition~\ref{prop-unique-sa}~(i).
Part (ii) follows analogously and part (iii) follows from part (i) and (ii).
\begin{prop}
If $\sum_{n=2}^\infty |a_n^{-1}|=\infty$ then $\Hh^+$ is uniquely self-adjoint.
\end{prop}

\begin{proof}
We have to show that both deficiency indices are zero.
Because of the left boundary condition, the only possible candidate for the kernel of $\Hh^+_{\max}-z$ is $\bigoplus_n \Psi^u_{z,n}$ which is an eigenvector if and only if $\sum_{n=1}^\infty \|\Psi^u_{z,n}\|^2\,<\,\infty$.
Hence, for the deficiency indices we have 
$$
d^+_\pm=0\quad\Leftrightarrow\quad \sum_{n=1}^\infty \|\Psi^u_{z,n}\|^2\,=\,\infty \quad \text{for $\pm\im(z)>0$}\;.
$$
So let us assume $\sum_{n=1}^\infty \|\Psi^u_{z,n}\|^2\,<\,\infty$ for some non-real $z$.
Note that
$$
|u_{z,n}|^2\,=\,|\Ups_n^* \Psi^u_{z,n}|^2\,\leq\,\|\Psi^u_{z,n}\|^2 \qtx{and} |\tilde u_{z,n}|^2\,=\,|\Phi_n^* \Psi^u_{z,n}|^2\,\leq\,\|\Psi^u_{z,n}\|^2\;.
$$
Using Lemma~\ref{lem-sum-est} this means for $n\geq 1$,
$$
\frac{|a_{n+1}|\,|u_{z,n+1} \,\tilde u_{z,n}|}{|\im(z)|}\,\geq\, \frac{|\im(a_{n+1}\,\overline{u}_{z,n+1} \,\tilde u_{z,n})}{|\im(z)|}\,\geq\,\|\Psi^u_{z,1}\|^2\,\geq\,|u_{z,1}|^2\,=\,1\;.
$$
Hence, 
$$
\frac{|\im(z)|}{|a_{n+1}|}\,\leq\,|u_{z,n+1}\,\tilde u_{z,n}| \quad\Rightarrow\quad \sum_{n=2}^\infty \frac{1}{|a_n|}\,\leq\,\frac1{2|\im(z)|} \sum_{n=1}^\infty \big(\|\Psi^u_{z,n}\|^2\,+\,\|\Psi^u_{z,n-1}\|^2\big)\,<\,\infty\;.
$$
Therefore, if $\sum_{n=1}^\infty |a_n^{-1}|=\infty $, then $\sum_{n=1}^\infty \|\Psi^u_{z,n}\|^2\,=\,\infty$ for any $z\not\in\RR$ and $\Hh^+_{\min}$ is essentially self-adjoint.
\end{proof}

Next, let us see how the 'limit point criterion' translates here and have a look at the Weyl circles.
For fixed $z\not \in (B_\infty\cup \RR)$ and $c$ varying in the real lines, the values $m_{n,c}(z)$ and $\tilde m_{n,c}(z)$ form a circle in the upper or lower half plane depending on the sign of $\im(z)$. 
Increasing $n$, the radius of these circles shrink and one either has a limit point or a limit circle. Whether it is limit point or limit circle does not depend on $z$ or the choice of $m$ or $\tilde m$.
Correspondingly, we say that $\Hh^+$ is either in the limit point or limit circle case. 
For Jacobi operators it is well known that the deficiency indices $d^\flat_+$ and $d^\flat_-$ are always equal and this dichotomy also classifies unique self-adjointness. 
Here, one needs some additional assumption. 

\begin{assumption}
 \begin{enumerate}
\item[{\rm (A1)}] For some non-real $z\not\in B_\infty$ the absolute values $|\det(T_{z,0,\pm n})|$ are uniformly bounded away from zero and infinity, 
i.e. for some constants $0<c<C<\infty$ and all $n$ we have $c<\det(T_{z,0,n})<C$.
 \item[{\rm (A2)}] All $\Ups_n,\;\Phi_n$ are real valued vectors and all $V_n$ are real symmetric matrices.
\item[{\rm (A3)}] For all $n$, $\Ups_n$ and $\Phi_n$ are co-linear, i.e. $\Phi_n = e^{i\theta_n} \Ups_n$ for some real phases $\theta_n$. 
\end{enumerate}
\end{assumption}
In fact, (A1) is the key assumption, (A2) and (A3) are special cases which imply (A1).

\begin{prop}\label{prop-limit-point}
\begin{enumerate}[{\rm (i)}]
\item The operator $\Hh^+$ is in the limit point case if and only if one of the deficiency indices $d^+_\pm$ is zero.
Similarly $\Hh^-$ is in the limit point case if and only if one of the deficiency indices $d^-_\pm$ is zero.
\item Assume that {\bf one} of the assumptions {\rm (A1)} or {\rm (A2)} or {\rm (A3)} listed above hold.
Then, $d^+_+=d^+_-,\,d^-_+=d^-_-$ and $d_+=d_-$. The operator $\Hh^\pm_{\min}$ is essentially self-adjoint if and only if $\Hh^\pm$ is in the limit point case,
$\Hh_{\min}$ is essentially self-adjoint if and only if both operators $\Hh^\pm$ are in the limit point case.
\end{enumerate}
\end{prop}

One may note that (A3) is particularly always satisfied for Jacobi operators. From now on we assume that $\Hh_{\min}$ is essentially self-adjoint.

\begin{proof}
Let
$
\smat{a&b\\c&d}\,\cdot\,x\,=\,\frac{ax+b}{cx+d}
$
denote the Möbius action.
From the equations \eqref{eq-def-m-fct} and definition of $x^{(n,c)}_z,\,\tilde x^{(n,c)}_z$ we get $T_{z,0,n}\cdot(a_1^2 m_{n,c}(z))=c$. Hence, using \eqref{eq-T-u-v} we get
\begin{equation}\label{eq-m}
m_{n,c}(z)\,=\,\frac{1}{a_1^2} \;T^{-1}_{z,0,n}\,\cdot\,c\,=\,
\frac{\tilde w_{z,n} c\,-\,a_{n+1} w_{z,n+1}}{ a_{n+1} u_{z,n+1}\,-\,\tilde u_{z,n} c}
\;.
\end{equation}
Taking $c\in \RR$ and $z\not\in \RR$ forms a circle in the upper (for $\im(z)>0$) or lower (for $\im(z)<0$) half-plane.
Moreover, for taking $\im(c)>0$ for $\im(z)>0$  (or $\im(c)<0$ for $\im(z)<0$) the Möbius action gives a point inside the circle.
In particular, as 
$$
T^{-1}_{z,0,n+1}\,\cdot\,c\,=\,T^{-1}_{z,0,n}\,\cdot\,(T^{-1}_{z,n+1}\,\cdot\,c)
$$
the $n+1$st circle lies strictly inside the $n$-th circle and the radius $r_{z,n}$ of the $n$-th circle is decreasing.
Depending on whether the radius converges to zero or not, we either have a limit point or a limit circle. The radius of the $n$-th circle satisfies
\begin{align*}
2\,r_{z,n}\,&=\,\sup_{c\in\RR} \left|\frac{\tilde w_{z,n} c\,-\, a_{n+1} w_{z,n+1}}{ a_{n+1} u_{z,n+1}\,-\,\tilde u_{z,n} c}\,+\,\frac{w_{z,n+1}}{u_{z,n+1}} \right|
\,=\,\sup_{c\in\RR} \left|\frac{\det(T_{z,0,n})\,c/a_{n+1}}{(a_{n+1}u_{z,n+1}\,-\,\tilde u_{z,n} c)\overline{u}_{z,n+1}}  \right| \notag \\
&=\,\sup_{c'\in \RR}\,\frac{|\det(T_{z,0,n})|}{\big|\;a_{n+1}\left( |u_{z,n+1}|^2 c'\,-\, \overline{u}_{z,n+1} \tilde u_{z,n}\,\right)\big|}\,=\,
\left|\frac{\det(T_{z,0,n})}{\im(a_{n+1} \overline{u}_{z,n+1} \tilde u_{z,n})}\right|
\end{align*}
In the last line it is not hard to see that the supremum is indeed achieved in the case where the denominator is purely imaginary as the imaginary part stays fixed whereas the real part varies through the whole of $\RR$.
Using Lemma~\ref{lem-bar-u} one can deduce that $r_{z,n}=r_{\bar z,n}$ which also follows from $m_{n,c}(\bar z)\,=\,\overline{m}_{n,c}(z)$ for real $c$.
\begin{equation}
 2\,r_{z,n}\,=\,\frac{\left|\det(T_{z,0,n})\right|}{|\im(z)|\sum_{k=1}^n \|\Psi^u_{z,k}\|^2}\,=\,
 \frac{\left|\det(T_{\bar z,0,n})\right|}{|\im(z)|\sum_{k=1}^n \|\Psi^u_{\bar z,k}\|^2}\,=\,2\,r_{\bar z,n}\;.
\end{equation}
Using the identities in Lemma~\ref{lem-bar-u} we get $|\det(T_{z,0,n})\,\det(T_{\bar z, 0,n})\,|=1$ and hence
\begin{equation}
 4\,|\im(z)|^2\,r_{z,n}^2\,=\,\frac{1}{\left[\sum_{k=1}^n \|\Psi^u_{z,k}\|^2\right]\left[\sum_{k=1}^n \|\Psi^u_{\bar z,k}\|^2\right]}
\end{equation}
and hence we are in the limit point case if and only if $\left[\sum_{k=1}^\infty \|\Psi^u_{z,k}\|^2\right]\left[\sum_{k=1}^\infty \|\Psi^u_{\bar z,k}\|^2\right]=\infty$ which is equivalent to one of the sums approaching infinity.
By the above arguments this is also equivalent to one of the deficiency indices $d^+_\pm$ being zero, showing part (i).

If for some $z$ the determinant $\det(T_{z,0,n})$ is uniformly bounded away from 0 and $\infty$, then the above relations show that $\sum_{n=1}^\infty \|\Psi^u_{z,n}\|^2=\infty \Leftrightarrow \sum_{n=1}^\infty \|\Psi^u_{\bar z,n}\|^2=\infty$ and therefore,
$d^+_+=d^+_-$. This shows part (ii) under additional assumption (A1).

If (A2) is true, i.e. $\Ups_n, \Phi_n, V_n$ are all real valued, then
$$
\beta_{z,n}=\Ups_n^* (V_n-z)^{-1} \Phi_n=\Ups_n^\top (V_n-z)^{-1} \Phi_n = \Phi_n^\top (V_n^\top-z)^{-1} \Ups_n=\Phi_n^*(V_n-z)^{-1} \Ups_n=\gamma_{z,n}.
$$
Hence, $\det(T_{z,n})=1$ and $\det(T_{z,0,n})=1$ for any $z$ and any $n$. Therefore, assumption (A1) is correct.

In case of (A3), $\Ups_n=e^{i\theta_n} \Phi_n$ we find $\gamma_{z,n}=e^{-2i\theta_n}\beta_{z,n}$ and $\det(T_{z,n})=e^{2i\theta_n}$. Thus, we get $|\det(T_{z,0,n})|=1$ for any $z$ and any $n$ giving again (A1).
\end{proof}

\subsection{Spectral measures}

In this section we assume that $\Hh$ (and hence also $\Hh^+$ and $\Hh^-$) are uniquely self-adjoint, in particular, we have the limit point case and we find uniformly for any sequence $(a_n)_n$ of real numbers
and fixed $z\not\in\RR$ that
\begin{align} \label{eq-wc-conv}
 m^+(z)&:=\langle P_1 \Ups_1; (\Hh^+-z)^{-1} P_1\Ups_1\rangle\,=\,\lim_{n\to\infty} m_{n,a_n}(z)\;
\end{align}
Indeed, since $\Dd^+_{\min}$ is a core, and $\Hh_{n,a_n}\Psi \to \Hh^+ \Psi$ for any $\Psi\in \Dd^+_{\min}$, we have strong resolvent convergence, i.e. $(\Hh_{n,a_n}-z)^{-1} \to (\Hh^+-z)^{-1}$ for $z\not \in \spec(\Hh^+)$, strongly.
Thus, the above equation holds for all $z$ not in the spectrum.

The uniformity of the convergence (independent of the sequence $a_n$) follows from the shrinking Weyl circle
for fixed $z\not\in (\RR\cup B_{1,\infty})$. 
In fact, one can replace $m_{n,a_n}(z)$ by any value inside the Weyl circle and formally take $a_n$ in the upper half plane ($\im(a_n)>0$)
using \eqref{eq-m} to define $m_{n,a_n}(z)$ in this case. 
Such a choice of $a_n$ corresponds to a spectral average in the sense of Herglotz functions. This means, introducing the spectral probability measures $\mu_{\Ups_1}^+$ and $\mu_{\Ups_1}^{(n,c)}$ of 
the operators $\Hh^+$ and $\Hh_{n,c}$ at the vector $P_1 \Ups_1$ one has 
$$
m_{n,c}(z)\,=\,\int_\RR (x-z)^{-1}\,\mu_{\Ups_1}^{(n,c)}(dx)\;,\quad m^+(z)\,=\,\int_\RR (x-z)^{-1}\,\mu_{\Ups_1}^+(dx)\;.
$$
Then, the measures $\mu^{(n,c)}_{\Ups_1}$ for $c$ in the upper half plane are averages over the $\mu^{(n,c)}_{\Ups_1}$ for $c$ along the real line. 
Thus, they define holomorphic functions for $\im(z)\neq 0$ (even though \eqref{eq-m} formally only can be used as a definition for $z\not\in B_{1,n}$).
Using the uniform bounds $|m(z)|\leq \frac{1}{|\im(z)|}$ for all these $m$-functions implying local equi-continuity, it is clear that the convergence
\eqref{eq-wc-conv} still holds for all $z\not\in\RR$ and any sequence $a_n$ in the upper half plane and the convergence is locally uniform in $z$.

\begin{proof}[Proof of Theorem~\ref{th:H+msr}] For simplicity we will use notations $\mu_{n,c}$ for $\mu^{(n,c)}_{\Ups_1}$ and
$\mu^+$ for $\mu_{\Ups_1}^{+}$.
Averaging $\mu_{n,c}$ over the standard Cauchy distribution leads to the choice $c=i$, $\mu_{n,i}$.
As the corresponding Herglotz functions converge in the upper half plane, the sequence $\mu_{n,i}$ converges weakly to the spectral measure $\mu^+$.
Taking a real value $z=\lambda\in\RR$, $\lambda\not \in B_{1,n}$ and using \eqref{eq-m} and the identities of 
Lemma~\ref{lem-bar-u} we find
$$
m_{n,i}(\lambda)=\frac{i-\tilde w_{\lambda,n} \overline{\tilde u}_{\lambda,n} - a_{n+1}^2 w_{\lambda,n+1} \bar u_{\lambda,n+1}}{a_{n+1}^2 |u_{\lambda,n+1}|^2\,+\,|\overline{\tilde u}_{\lambda,n}|^2}\;.
$$
In particular, for $\lambda\not\in B_{1,n}$, $m_{n,i}(z)$ extends holomorphically to $\lambda$ and the corresponding measure $\mu_{n,i}$ is absolutely continuous on $\RR\setminus B_{1,n}$.
Now, note that for real energies all entries of the transfer matrix $T_{\lambda,n}$ have the same complex phase as $\beta_{\lambda,n} T_{\lambda,n}$ is real valued. 
Hence, also all entries of the products $T_{\lambda,0,n}$ have the same phase. By the representation \eqref{eq-T-u-v} this means that $\tilde w_{\lambda,n} \overline{\tilde u}_{\lambda,n}$ and 
$w_{\lambda,n+1} \bar u_{\lambda,n+1}$ are real. Hence,
$$
\im(m_{n,i}(\lambda))=\frac1{a_{n+1}^2 |u_{\lambda,n+1}|^2\,+\,|\overline{\tilde u}_{\lambda,n}|^2}=\left\| T_{\lambda,0,n} \pmat{a_1 \\ 0} \right\|^{-2}
$$
and by general theory of Stielties transforms of measures, 
$$
\mu_{n,i}(d\lambda)\,=\,\pi^{-1} \left\| T_{\lambda,0,n} \smat{a_1 \\ 0} \right\|^{-2}{\rm d}\lambda\,+\,\nu_{n}(d\lambda)
$$
where $d\lambda$ denotes the Lebesgue measure on the real line and $\nu_{n}(d\lambda)$ is a positive point measure supported on the finite set $B_{1,n}\cap \RR$.
Clearly, this measure corresponds to values $\lambda$ which are eigenvalues of $\Hh_{n,c}$ for $c \in I$ where $I\subset \RR$ has positive measure.
By rank one perturbation arguments at $\phi_n$, this must correspond to some eigenfunction orthogonal to $P_n\phi_n$ which in fact is an eigenfunction $\Psi$ of $\Hh_{n,c}$ for all $c$
satisfying $\Phi_n^*\Psi_n=0=\Ups_n^*\Psi_n$ and it leads to a compactly supported eigenfunction of $\Hh^+$ as classified in Theorem~\ref{th-compact-eig}.
Vice versa, it is easy to see that any fixed compactly supported eigenfunction $\Psi$ of $\Hh^+$ there exists $n_0$ (upper bound of support) such that 
$\Psi$ is an eigenfunction of $\Hh_{n,c}$ for all $c\in\RR$ if and only if $n>n_0$. 

Therefore, the positive point measures $\nu_{n}(d\lambda)$ are increasing in $n$ and upper bounded (total mass is bounded by one) and hence converge weakly to a point measure $\nu^+$ supported on the countable set $B_{1,\infty}$.
This proves part (i).

Let us now generally define $m^+_\Psi(z):=\langle \Psi;(\Hh^+-z)^{-1} \Psi\rangle$.
From Proposition~\ref{prop:green} and strong resolvent convergence we find for $\varphi\in\VV_n\subset \ell^2(S_n)$ and $z\not\in B_{1,n} \cup \RR$ that
\begin{align}
 m^+_{P_n\varphi}(z)\,&=\,m^+(z)\,\varphi^*\Psi^u_{z,n} (\Psi^u_{\bar z, n})^*\varphi\, +\,\varphi^*(V_n-z)^{-1} \varphi \,+ \notag \\
 &\varphi^* \Phi_{z,n}^v (\Psi^u_{\bar z,n})^* \varphi\,+\,\varphi^* \Ups_{z,n}^u (\Psi^w_{\bar z,n})^* \varphi\,.
\end{align}
As all entries of $T_{\lambda,n}$ have the same complex phase for real $\lambda\in\RR\setminus B_{1,n}$ we find $w_{\lambda,n+1} \bar u_{\lambda,n+1} \in \RR$ and $\tilde w_{\lambda,n} \overline{\tilde u}_{\lambda,n} \in \RR$.
Using the definition \eqref{eq-def-psi_z} this leads to 
$$
\im\left(\varphi^* \Phi_{\lambda,n}^v (\Psi^u_{\lambda,n})^* \varphi\,+\,\varphi^* \Ups_{\lambda,n}^u (\Psi^w_{\lambda,n})^* \varphi \right)\,=\,0
$$
for $\lambda\in \RR \setminus (B_{1,n} \cup \spec(V_n|\VV_n))$. 
Now, the values $\lambda\in \spec(V_n|\VV_n)$ where $T_{\lambda,n}$ is defined are simple eigenvalues with eigenvector ${\bfa}\in\VV_n$ which has overlap to $\Ups_n$ and $\Phi_n$.
A simple calculation shows that the singularities of
$$
\varphi^*(V_n-z)^{-1} \varphi\,+\,\varphi^* \Phi_{z,n}^v (\Psi^u_{\bar z,n})^* \varphi\,+\,\varphi^* \Ups_{z,n}^u (\Psi^w_{\bar z,n})^* \varphi
$$
cancel at such $z=\lambda$. Thus, for $z\to \lambda$ with $\lambda\in\RR\setminus B_{1,n}$ and $\im(z)>0$ we find
$$
\im\left(m^+_{P_n\varphi}(z)\right)\,-\,\im\left(m^+(z)\,\varphi^*\Psi^u_{z,n} (\Psi^u_{\bar z, n})^*\varphi \right)\,\to\,0\;.
$$
Similar arguments as in \cite[Lemma~A.2~(ii)]{Sa3} therefore give part (ii).

\vspace{.2cm}

Before considering the full operator, we want to have a look at some identities again.
But this time we add a right boundary condition and let $\Hh_{b,N,c}=\Hh_{N,c}-b\,a_1^2\,P_1\Ups_1 \Ups_1^* P_1^*$.
Fixing $b$ this would amount to the boundary condition $b a_1 \Ups_1^*\Psi_1=\Phi_0^* \Psi_0$ or $\tilde x_0=b a_1 x_1$. Hence, the base solution $(u_z,\tilde u_z)$ has to be replaced by
$(u_z+b a_1^2 w_z,\tilde u_z + b a_1^2 w_z)$ which by \eqref{eq-T-u-v} amounts to replacing $T_{z,n,0}$ by $T_{z,n,0}\smat{1 & 0 \\ b & 1}$.
Alternatively, redefining $V_1$ to $V_1-b a_1^2 \Ups_1 \Ups_1^*$ and resolvent identities also give the change $T_{z,1} \to T_{z,1} \smat{1 & 0 \\ b & 1}$.
Therefore, fixing $b$ and spectral averaging over $c$ as above\footnote{cf. taking $c=i$ in the Möbius transform for the $m$-function} for $\Hh_{b,N,c}$ leads to the measures $\mu^{(b,N)}_\psi$.
Similar as above, we then have
\begin{equation}\label{prTh3-1}
a_1^2\,\mu^{(b,N)}_{\Ups_1}(d\lambda)\,=\,\frac{1}{\pi}\, \| T_{\lambda,0,n} \smat{1 & 0 \\ b & 1} \smat{1\\0}\,\|^{-2}\,d\lambda\,+\,\nu_{b,N}(d\lambda)
\end{equation}
and for $\lambda \not \in B_{1,N}$, $\varphi\in \VV_k$
\begin{equation}\label{prTh3-2}
\mu_{P_k\varphi}^{(b,N)}(d\lambda)\,=\,|\varphi^* \psi^x_{\lambda,k} |^2\,\mu^{(b,N)}_{\Ups_1}(d\lambda)\;
\end{equation}
where $\psi^x_{\lambda,k}$ belongs to the solution $(x_z,\tilde x_z)$ to the transfer matrix equation with $x_{z,1}=1$ and $\tilde x_{z,0}=b a_1$, i.e. $x_z=u_z+b a_1^2 w_z$, $\psi^x_{z,n}=\psi^u_{z,n}+b a_1^2 \psi^w_{z,n}$.
In particular, we find for $\lambda\not\in B_\infty$. $x,y \in \RR$
\begin{equation}\label{prTh3-3}
\left(x^2 \mu^{(b,N)}_{\Ups_{l+1}}\,+\,y^2 \mu^{(b,N)}_{\Phi_{l}}\right)(d\lambda)\,=\,\left\|\smat{x a_{l+1}^{-1} & \\ & y}T_{\lambda,0,l} \smat{a_1 \\ ba_1} \right\|^2\,\mu^{(b,N)}_{\Ups_1}(d\lambda)
\end{equation}
Let us now shift the indices and consider $\Hh_{[m,n]}$, i.e. $\Hh$ restricted to $\bigoplus_{k=m+1}^n \ell^2(S_k)$ with left and right boundary conditions $(b,c)$ as for the operator $\Hh_{b,N,c}$. 
The corresponding spectral averages over the right boundary are denoted by $\mu^{b,(m,n)}_\psi$. 
Then, shifting the indices we get from \eqref{prTh3-1} and \eqref{prTh3-3} that for $m<0<n$,
\begin{align}
& \left(a_1^2 \mu^{b,(m,n)}_{\Ups_1}+\mu^{b,(m,n)}_{\Phi_0}-\nu^{(m,n)}\right)(d\lambda)\,=\,  \frac{d\lambda}{\pi} \; \frac{\|T_{\lambda,m,0}\smat{1\\b}\|^2}{\|  T_{\lambda,0,n} T_{\lambda,m,0} \smat{1\\b}\|^2}\,.
\end{align}
The positive point measure $\nu^{(m,n)}(d\lambda)$ is supported on $B_{m,n}\cap \RR$ and induced by the boundary-condition-independent eigenfunctions
as in Theorem~\ref{th-compact-eig}.
Averaging over the left boundary leads to
\begin{align}
\int_\RR \frac{db}{\pi(1+b^2)} \frac{\|T_{\lambda,m,0}\smat{1\\b}\|^2}{\|  T_{\lambda,0,n} T_{\lambda,m,0} \smat{1\\b}\|^2}
 &=\int_\RR \frac{db}{\pi(1+b^2)} \frac{\|\smat{1\\b}\|^4}{\|T_{\lambda,0,n} \smat{1 \\ b}\|^2 \|T_{\lambda,m,0}^{-1}\smat{1\\b}\|^2} \\
 &= \int_0^\pi d\theta \left\| T_{\lambda,0,n} \smat{\cos(\theta) \\ \sin(\theta)} \right\|^{-2} \left\|T_{\lambda,m,0}^{-1}\smat{\cos(\theta)\\ \sin(\theta)}\right\|^{-2}
\end{align}
The equations come from changes of variables, $T_{\lambda,m,0} \smat{1 \\ b} = C \smat{1 \\ b'} = C \smat{1 \\ \tan(\theta)}$. 
Here one should note that for real energies $z=\lambda\in\RR$ the complex phases of all entries of the transfer matrix $T_{\lambda,m,0}$ are the same. Hence, $T_{\lambda,m,0}$ acts naturally on the real projective space
identifying $\smat{a\\b}$ with $\frac{a}{b} \in \RR \cup\{\infty\}$.
Using resolvent convergence, $m\to -\infty$ and $n\to\infty$ gives \eqref{eq-mu1} proving part (iii).

For the analogue of \eqref{prTh3-2} we need to take the solutions $(x_z, \tilde x_z)$ with 
$x_{z,m+1}=1$  and $\tilde x_{z,m}=b a_{m+1}$. Then,
$\smat{a_1 x_{z,1} \\ \tilde x_{z,0}}\,=\,T_{z,m,0} \smat{a_{m+1} \\ b a_{m+1}}$ and $x_z=x_{z,1} u_z+\tilde x_{z,0} a_1 w_z$, implying 
$$
\varphi^* \psi_{z,k}^w\,=\, \ovr{v}_{\varphi,z,k}^* \pmat{a_1^{-1} \\ & a_1} T_{z,m,0} \pmat{a_{m+1} \\ b a_{m+1}}\qtx{with} 
\ovr{v}_{\varphi,z,k}\,:=\,\pmat{(\Psi^u_{z,k})^*\varphi \\ (\Psi^w_{z,k})^*\varphi }
$$
Thus, \eqref{prTh3-2} and \eqref{prTh3-3} applied to this situation give for $\varphi\in \VV_k\subset \ell^2(S_k)$, $\lambda\in \RR \setminus B_{m,n}$, $m<0<n$, $m<k<n$, the Radon-Nikodym derivative
$$
\frac{{\rm d}\mu^{b,(m,n)}_{P_k \varphi}(\lambda)}{{\rm d} (\mu^{b,(m,n)}_{\Ups_1}+a_1^2\mu^{b,(m,n)}_{\Phi_0})(\lambda)}\,=\,
\frac{\left| \ovr{v}_{\varphi,\lambda,k}^* \smat{a_1^{-1} \\ & a_1}T_{\lambda,m,0} \smat{1 \\ b } \right|^2}{\left\|\smat{a_1^{-1} \\ & a_1}T_{\lambda,m,0}\smat{1\\b} \right\|^2}\,\leq\,\left\| \ovr{v}_{\varphi,\lambda,k}\right\|^2\,.
$$
Taking $m\to-\infty$ and $n\to\infty$ shows part (iv).
 \end{proof}


\section{Application to random operators}

\subsection{Proof of Theorem~\ref{th-random}}

We start with the folowing important lemma following the arguments of \cite{KLS}.

\begin{lemma}\label{lemma-T4-bound}
Let $T\in\exp(i\RR)\SL(2,\RR)$ with $|\Tr(T)<2|$ and let $W_n\in \Mat(2,\CC)$ be a sequence of independent random matrices such that
$$
\sum_{n=1}^{\infty} \big(\left\|\EE(W_n)\right\|\,+\,\EE(\|W_n\|^2+\|W_n\|^4)\, \big)\,=\,C\,<\,\infty,
$$
then,
$$
\sup_{n\in\NN} \,\EE\,\left(\left\|\,\prod_{k=1}^n (T+W_k)\,\right\|^4\right)\,<(2f(T))^4\,\exp(8f(T)C)\,<\,\infty
$$
where $f(T)$ is some continuous function in $T$ on the set $\{T\in \SL(2,\RR)\,:\, |\Tr(T)|<2\}$.
\end{lemma}
Here, we define products to the left, i.e. $\prod_{k=1}^n T_k= T_n T_{n-1}\cdots T_1$, but for the validity of the lemma this would not be important.

\begin{proof}
Such $T\in\exp(i\RR)\SL(2,\RR)$ with $|\Tr(T)<2|$ is equivalent to a unitary matrix, i.e. there is $B\in \GL(2,\RR)$ such that $B^{-1}TB=U\in {\rm U}(2)$.
Locally, $B$ can be chosen to depend continuously on $T$ and
therefore we find some continuous function $f(T)$ with $f(T)\geq \|B\| \|B^{-1}\|$. So conjugations by $B$ or $B^{-1}$ increases the norm at most by a factor $f(T)$
and we may assume $T=U$, i.e. $T^*T=\one$, and $f(T)=1$.

For any fixed unit vector $\bfv \in \CC^2$ let 
$$
r_n(\bfv)=\left\|\left(\prod_{k=1}^n  T+W_k\right )\bfv\right\|\,.
$$
Then, there is a unit vector $\bfw$ such that
$$
r_n^2(\bfv)\,=\,r_{n-1}^2(\bfv)\,\bfw^*( T+W_n)^*(T+W_n)\bfw\,=\,r_{n-1}^2\left(1+2\re(\bfw^*T^*W_n\bfw) +\|W_n\bfw\|^2\right)
$$
implying
\begin{align*}
r_n^4(\bfv)\,&=\leq\,r_{n-1}^4(\bfv)\left(1+4\re(\bfw^*T^*W_n\bfw)+6\|W_n\|^2+4\|W_n\|^3+\|W_n\|^4\right)
\end{align*}
where we use $|\re(\bfw^*T^*W_n\bfw)|\leq \|W_n\|$.
Now, $\|\EE(W_n^*T+T^*W_n)\|\leq 2 \|\EE(T^*W_n)\|=2\|\EE(W_n)\|$ and $2\EE(\|W_n\|^3)\leq \EE(\|W_n\|^2+\|W_n\|^4)$, hence
$$
\EE(r_n^4(\bfv))\,\leq\, \EE(r_{n-1}^4(\bfv))\,\left(1+4\| \EE(W_n) \| + \EE(8\| W_n\|^2+3\|W_n\|^4) \right)
$$
Here we use that $r_{n-1}$ is independent of $W_n$ as $W_n$ is independent of $W_1,\,W_2,\ldots,W_{n-1}$.
This implies
$$
\EE(r_n^4(\bfv))\,\leq\, \prod_{k=1}^n \exp\left(8(\|\EE(W_n)\|\,+\,\EE(\|W_n\|^2+\|W_n\|^4))\right) \leq \exp(8C)
$$
for any $n$ and any unit vector $\bfv$. \\
If $\bfv_1, \bfv_2$ is a basis of $\CC^2$, then
$\|M\|\leq \|M\bfv_1\|+\|M\bfv_2\|$ and therefore,
$$
\,\EE\left( \left\|\,\prod_{k=1}^n (T+W_k)\,\right\|^4\right)\,\leq \,2^4 \exp(8C)\;.
$$
Which proves the special case $T=U$ unitary. The general case follows by basis change as argued above.
\end{proof}

Now we can prove Theorem~\ref{th-random}:

\begin{proof}[Proof of Theorem~\ref{th-random}]
Under the assumptions we have $T_{\lambda,n}=T(\lambda)+W_{\lambda,n}$ with $T(\lambda)$ and $W_{\lambda,n}$ satisfying the conditions of Lemma~\ref{lemma-T4-bound}. Moreover, the bounds for this condition are uniform in $\lambda\in[a,b]\subset I$ for any compact intervals $[a,b]\in I$. Thus, Lemma~\ref{lemma-T4-bound} gives a bound on $\EE(\|T_{\lambda,0,n}\|^4)$ uniformly in $n\in\NN$ and $\lambda\in[a,b]\subset I$
 in compact intervals $[a,b]\subset I$. Therefore, by Fatou's Lemma and Fubini,
 $$
 \EE\liminf_{n\to\infty} \int_a^b \|T_{\lambda,0,n}\|^4 \,d\lambda\,\leq\,\liminf_{n\to\infty} \int_a^b \EE\left\|T_{\lambda,0,n}\right\|^4\,d\lambda\,<\,\infty\;.
 $$
Hence,
$$
\liminf_{n\to\infty} \int_a^b \|T_{\lambda,0,n}\|^4 \,d\lambda\,<\,\infty
$$
with probability one and the spectrum of $\Hh_\omega$ is absolutely continuous in $(a,b)$ with probability one and $[a,b]$ is in the spectrum with probability one by Theorem~\ref{th:ac-spec}.
The open set $I$ can be written as a countable union of intervals $(a,b)$ such that $[a,b]\subset I$.
\end{proof}

\subsection{Application to partial antitrees}

For the concrete examples we will show that the assumptions of Theorem~\ref{th-random} are satisfied in specific intervals.
The main point will be to treat the random variables $\alpha_{\lambda,n},\, \delta_{\lambda,n}$ and\footnote{$\beta_{\lambda,n}=\gamma_{\lambda,n}$ follows as we consider real valued random operators and $\lambda\in\RR$}
 $\beta_{\lambda,n}=\gamma_{\lambda,n}$. We have to find almost sure limiting expressions and bound
variances up to sufficient order.
To get these bounds we will use the following lemmata.

\begin{lemma}\label{lem-averages}
For any $\lambda\in I$ let be given a real random variable $X_\lambda$ with $\EE|X_\lambda|^{2K}<C$ uniformly in $\lambda\in I$, $K\in\NN$. 
Let $(X_{\lambda,j,n})_{j,n}$ be independent, identically distributed copies of $X_\lambda$ let $(s_n)_n$ be some sequence with $s_n\to\infty$ for $n\to\infty$.
Then  we find
$$
\frac{1}{s_n}\sum_{j=1}^{s_n} X_{\lambda,j,n}\xrightarrow[2K]{\Oo(s_n^{-1/2})} \EE(X_\lambda) \qtx{uniformly for} \lambda\in I\;.
$$
\end{lemma}
\begin{proof}
Without loss of generality we may assume $\EE(X_\lambda)=0$. Furthermore we ignore the dependence on $\lambda$ and let $X_{j}$ be independent identically distributed random variables 
with $\EE(X_j)=0$, $\EE|X|^K<C$. Uniformity in $\lambda$ is established by the explicit dependence of the bounds on $C$.
Let $Y=\frac{1}{m}  \sum_{j=1}^m X_j$. Then
$\EE(Y)=0$ and letting $c_{k,l}$ denote the number of partitions of $\{1,\ldots,2k\}$ into $l$ sets where each of the sets has at lest two elements, then we have for $k\leq K$,
\begin{align*}
\EE (Y^{2k}) &= \frac{1}{m^{2k}} \sum_{i_1,\ldots,i_{2k}=1}^m \EE( X_{i_1} X_{i_2}\cdots X_{i_{2k}}) \\ &\leq \sum_{l=1}^k \frac{c_{k,l}m^{l}}{m^{2k}} \EE|X_1|^{2k}
\leq \frac{1}{m^{k}} (1+C) \sum_{l=1}^k c_{k,l} \,.
\end{align*}
The first inequality we used that $\EE(X_{i_1} \cdots X_{i_{2k}})=0$ if one of the indices does not appear twice and $\EE|X_{i_1} \cdots X_{i_{2k}}|\leq \EE(\sum |X_{i_j}|/k)^k\leq \EE|X_1|^{2k}$ following from the geometric - arithmetic mean inequality and convexity of
$x\mapsto x^k$. Note that this implies $\EE|Y|^{2k}\leq C_k m^{-k}$ with $C_k$ only depending on $k$ and the constant $C$.
Furthermore $\EE|Y|^{2k-1}\leq \sqrt{\EE|Y|^{2k-2}} \sqrt{\EE|Y|^{2k}} \leq \sqrt{C_{k-1} C_k} m^{-k+1/2}$.  
Using $m=s_n$ yields the claim.
\end{proof}

\begin{lemma}\label{lem-ot-close}
Let $(t_n)_{n\in\NN}$ be some sequence of positive elements with $t_n\to 0$ for $n\to\infty$. Let $X^{(j)}_{\lambda,n}$ be families of real random variables with
$X^{(j)}_{\lambda,n}\xrightarrow[2K]{\Oo(t_n)} X^{(j)}_\lambda$ uniformly in $\lambda\in I$, $j=1,\ldots,k$, with $K\geq 1$. Then, the following hold:
\begin{enumerate}[{\rm (i)}]
\item $X^{(1)}_{\lambda,n}+X^{(2)}_{\lambda,n}\xrightarrow[2K]{\Oo(t_n)} X^{(1)}_{\lambda}+X^{(2)}_{\lambda}$ uniformly in $\lambda$
\item If the support of the distribution of the random variables $X^{(1)}_{\lambda,n}$ is uniformly bounded away from zero in $\lambda$ and $n$, 
then $1/X_{\lambda,n}^{(1)} \xrightarrow[2K]{\Oo(t_n)} 1/X_\lambda^{(1)}$ uniformly in $\lambda$.
\item Let $f\in C^2(\RR^k,\RR)$ and assume that $\nabla f(X_\lambda^{(1)},\ldots,X_\lambda^{(k)})$ and the Hessian of $f$ are uniformly bounded, then 
$$f(X_{\lambda,n}^{(1)},X_{\lambda,n}^{(2)},\ldots,X_{\lambda,n}^{(k)}) \xrightarrow[K]{\Oo(t_n)} f(X_{\lambda}^{(1)},X_{\lambda}^{(2)},\ldots,X_{\lambda}^{(k)})$$ uniformly in $\lambda$.\\
In particular,  if $X^{(j)}_{\lambda}$ are uniformly bounded then $X^{(1)}_{\lambda,n}X^{(2)}_{\lambda,n}\xrightarrow[K]{\Oo(t_n)} X^{(1)}_\lambda X^{(2)}_\lambda$ uniformly in $\lambda$.
\end{enumerate}
\end{lemma}
\begin{proof} For the proof we may neglect the dependence on $\lambda$. The uniformity part is settled by getting the new bounds  
as concrete expressions from the old ones.
Thus, let $X^{(j)}_n$ be random variables and $X^{(j)}\in\RR$ such that
$$
\|\EE(\vec X_n)-\vec X\|\,<\,C_0 t_n^2\,,\quad\EE\|\vec X_n-\vec X\|^k\,<\,C_k t_n^k
$$
for $k=1,\ldots,2K$, where $\vec X= (X^{(1)},\ldots,X^{(k)})$
Then
$$
|\EE(X^{(1)}_n+X^{(2)}_n)-(X^{(1)}+X^{(2)})|< 2 C_0 t_n^2
$$
and
$$
\EE|(X^{(1)}_n+X^{(2)}_n-X^{(1)}-X^{(2)})|^k\leq\sum_{j=0}^k \binom{k}{j} |X_n^{(1)}-X^{(1)}|^j |X^{(2)}_n-X^{(2)}|^{k-j}<\sum_{j=0}^k \binom{k}{j} C_{k-j} C_j \,t_n^k
$$
for any $k=1,\ldots, 2K$.\\
For (ii)  let us use the notation $X_n=X^{(1)}_n$ and $X=X^{(1)}$ and
note that with $\inf_n \supp(|X_n|)=c>0$ we find
$$
\EE\left( \left|\frac{1}{X_n}-\frac{1}{X}\right|^k\right)\,=\,\EE\left(\frac{|X-X_n|^k}{|X_n X|^k} \right)\,\leq\,c^{-2k}\,\EE(|X-X_n|^k)\,\leq\,c^{-2k} C_k t_n^k 
$$
for $k=1,\ldots, 2K$. To get the well-balanced property, note that
$$
\left|\EE\left(\frac{1}{X_n}-\frac{1}{X} \right)\,\right|\,\leq\,
\frac{|\EE(X-X_n)|}{|X|^2}\,+\,\left|\EE\left(\frac{(X-X_n)^2}{X^2 X_n}\right)\right|\,\leq\, (c^{-2} C_0+c^{-3} C_2)t_n^2\,.
$$
Finally for (iii) note that a Taylor expansion with remainder term yields
$$
f(\vec X_n)-f(\vec X)\,=\,
\nabla f(\vec X)\,\cdot\,(\vec{X}_n-\vec{X})\,+\,(\vec{X}_n-\vec{X})\,\cdot\,A(\vec{X}_n)\,(\vec{X}_n-\vec{X})
$$ 
where $\|A(\vec{X}_n)\|<C$ is uniformly bounded. Hence,
$$
\left|\EE\left[(\vec{X}_n-\vec{X})\,\cdot\,A(\vec{X}_n)\,(\vec{X}_n-\vec{X})\right]\right|
\,\leq\, C \EE\|\vec{X}_n-\vec{X} \|^2\,\leq\, C C_2 t_n^2
$$
and
$$
\EE\left|(\vec{X}_n-\vec{X})\,\cdot\,A(\vec{X}_n)\,(\vec{X}_n-\vec{X})\right|^k\,\leq\,
C^k\,\EE\|\vec{X}_n-\vec{X}\|^{2k}\,\leq\, C C_{2k} t_n^{2k}\;.
$$
for $k=1,\ldots K$.
This implies $(\vec{X}_n-\vec{X})\,\cdot\,A(\vec{X}_n)\,(\vec{X}_n-\vec{X})\xrightarrow[K]{\Oo(t_n)} 0 $.
Clearly, $\nabla f(\vec X)\cdot \vec X_n \xrightarrow[2K]{\Oo(t_n)}\nabla f(\vec X)\cdot \vec X$ with constants
$\| f(\vec X)\| C_k$. Using part (i) now gives the result.
\end{proof}

In our specific models we always have $a_n=-1$ and the entries of $T_{\lambda,n}$ consist of the {\bf real}\footnote{in these examples all matrices are real for $\lambda\in\RR$} random variables
$\beta_{\lambda,n}^{-1}$, $\beta_{\lambda,n}^{-1} \alpha_{\lambda,n}$,
$\beta_{\lambda,n}^{-1}\delta_{\lambda,n}$, $\gamma_{\lambda,n}=\beta_{\lambda,n}$ and $\alpha_{\lambda,n} \delta_{\lambda,n} \beta_{\lambda,n}^{-1}$.

To get the bounds up to 4 moments for $T_{\lambda,n}$ we will show the well-balanced closeness up to any moments of $\alpha_{\lambda,n}$, 
$\beta_{\lambda,n}$ and $\delta_{\lambda,n}$ to $\alpha_\lambda,\,\beta_{\lambda}$ and $\delta_\lambda$, respectively. 
Furthermore, by their definition within the considered regions for $\lambda$ the distribution of $\beta_{\lambda,n}$ is uniformly (in $n$ and on compact sets for $\lambda$) bounded away from zero.

\begin{proof}[Proof of Theorem~\ref{theo-Anderson} part (i), stretched antitree $\Sb_\bfs$]

Let us consider the terms $\alpha_{\lambda,n}$, $\beta_{\lambda,n}$ and $\delta_{\lambda,n}$ in more detail.
Recall, we defined the unit vectors $\varphi_n= r_n^{-1/2} (1,1,\ldots,1)^\top \in \CC^{r_n}\cong \ell^2(R_n)$.
We have $S_n=R_{3n-2}\sqcup R_{3n}$ (recall that $R_{3n-1}=\emptyset$ in this case) and 
$\#(R_{3n-2})=\#(R_{3n})=s_n$ are sets of the same size and 
we find that $\Hh_\omega$ is an operator as in \eqref{eq-Hh} with
$$
\Upsilon_n=\pmat{\varphi_{3n-2} \\ \nul}\,,\quad \Phi_n=\pmat{\nul \\ \varphi_{3n}}\,,\quad V_n=\pmat{\bfv_{3n-2} & \one \\ \one & \bfv_{3n}}\;.
$$
Here, $\bfv_n=\bfv_{n}(\omega)=\diag((\omega_{x})_{x\in R_{n}})$ are random diagonal matrices, $\omega_x$ for $x\in\GG$ is the random i.i.d. potential on $\Sb_\bfs$.
 
It is easy to see that $\delta_{\lambda,n}$ and $\alpha_{\lambda,n}$ have the same distribution.
For the formulas we will order the points in $R_n$ and use the notations 
$\bfv_{3n-2}=\diag(\omega_{n,1},\ldots,\omega_{n,s_n})$ and $\bfv_{3n}=\diag(\omega'_{n,1},\ldots,\omega'_{n,s_n})$.
Standard Schur complement formulas give
\begin{align}
\alpha_{\lambda,n}\,&=\,\varphi_{2n-1}^*(\bfv_{2n-1}-\lambda- (\bfv_{2n}-\lambda)^{-1})^{-1} \varphi_{2n-1}\,=\,
 \frac{1}{s_n}  \sum_{j=1}^{s_n} \frac{\omega'_{n,j}-\lambda}{(\omega_{n,j}-\lambda) (\omega'_{n,j}-\lambda)-1} \\
\beta_{\lambda,n}\,&=\, \varphi_{2n-1}^*((\bfv_{2n}-\lambda)(\bfv_{2n-1}-\lambda)-\one )^{-1} \varphi_{2n} \,=\,
\frac{1}{s_n} \sum_{j=1}^{s_n} \frac{1}{(\omega_{n,j}-\lambda) (\omega'_{n,j}-\lambda)-1}
\end{align}
Using Lemma~\ref{lem-averages} we find
$$
\alpha_{\lambda,n} \xrightarrow{\Oo(s_n^{-1/2})} \alpha^\Sb_\lambda,\quad
\beta_{\lambda,n} \xrightarrow{\Oo(s_n^{-1/2})} \beta_\lambda^\Sb
$$
uniform for $\lambda$ in compact subsets of $I_{\Sb,\nu}$. Using Lemma~\ref{lem-ot-close} it follows 
$$
T_{\lambda,n} \xrightarrow{\Oo(s_n^{-1/2})} T_\lambda=-\pmat{(\beta_\lambda^\Sb)^{-1} &  -\alpha^\Sb_\lambda / \beta^\Sb_\lambda \\ \alpha^\Sb_{\lambda} / \beta^\Sb_\lambda & \beta^\Sb_\lambda - (\alpha_\lambda^\Sb)^2 / \beta^\Sb_{\lambda}}
$$
uniform for $\lambda$ in compact subsets of $I_{\Sb,\nu}$. 
We see that for $\lambda \in I_{\Sb,\nu}$ we have $|\Tr T_\lambda |<2$. Hence, by Theorem~\ref{th-random} it follows that with probability one
$I_{\Sb,\nu}\subset \spec(\Hh_\omega)$ and that the spectrum of $\Hh_\omega$ is purely absolutely continuous in $I_{\Sb,\nu}$ if
$\sum_{n=1}^\infty s_n^{-1} < \infty$.
\end{proof} 
 
\begin{proof}[Proof of Theorem~\ref{theo-Anderson} part (ii), graph $\Ab_{\bfr}(\vec{k},O,\bfa)$] 
In this case we have
$$
\Upsilon_n=\pmat{\varphi_{3n-2} \\ \nul \\ \nul}\;,\quad \Phi_n=\pmat{\nul \\ \nul \\ \varphi_{3n}}\;,\quad
V_n=O_n\bfa O_n^*+\Vv_n
$$
where $\Vv_n=\diag(\omega_x)_{x\in S_n}$, where $S_n=R_{3n-2}\sqcup R_{3n-1}\sqcup R_{3n}$.
To simplify notation let 
$$
S_{n,i}=R_{3n-2,i}\;; \quad S_{n,j+k_1}=R_{3n-1,j}\;,\quad S_{n,l+k_1+k_2}=R_{3n,l}
$$
for $i=1,\ldots, k_1;\;j=1,\ldots,k_2,\;l=1,\ldots k_3$.
For $\lambda\in I_0$ let us define the diagonal $k\times k$ matrix $\hat \Vv_{\lambda}$ by
$$
\hat \Vv_{\lambda,n}^{-1}\,=\,\diag\left(\frac{1}{\# (S_{n,j})}\sum_{x\in S_{n,j}} \frac{1}{\omega_x-\lambda}\right)_{j=1}^k
$$
For calculating the transfer matrix entries we use the following identity
\begin{align}
(V_n-\lambda)^{-1}&=(\Vv_n-\lambda+O_n\bfa O_n^*)^{-1}\notag \\
&=(\Vv_n-\lambda)^{-1}\left(\one-O_n[\bfa^{-1}+O_n^* (\Vv_n-\lambda)^{-1}O_n]^{-1} O_n^* (\Vv_n-\lambda)^{-1}\right) 
\label{eq-exp-inverse}
\end{align}
One may note that the inverse within the bracket $[\bfa^{-1}+O_n^* (\Vv_n-\lambda)^{-1}O_n]^{-1}$ is an inverse of $k\times k$ matrices, rather than $\#(S_n)\times \#(S_n)$ matrices.
By the setup and definition of $O_n$ we find
$$
O_n^* (\Vv_n-\lambda)^{-1} O_n = O^* \hat \Vv_{\lambda,n} O,\quad \Upsilon_n^* (\Vv_n-\lambda)^{-1} O_n= 
\Upsilon^*\hat\Vv^{-1}_{\lambda,n} O 
$$
$$
\Phi^* (\Vv_n-\lambda)^{-1} O_n=\Phi^* \hat \Vv^{-1}_{\lambda,n}
\qtx{and}
\Upsilon_n^* (\Vv_n-\lambda)^{-1} \Phi_n=\Upsilon^* \hat\Vv_{\lambda,n}^{-1} \Phi\;.
$$
Using this with \eqref{eq-exp-inverse} we see that for calculating $\alpha_{\lambda,n}, \beta_{\lambda,n}, \gamma_{\lambda,n}$ 
 on the right hand side of \eqref{eq-exp-inverse} we can replace $\Upsilon_n, \Phi_n, O_n, (\Vv_n-\lambda)$ by
 $\Upsilon, \Phi, O$ and $\hat \Vv_{\lambda,n}$. Using the same expansion as in \eqref{eq-exp-inverse} this finally gives
 $$
 \pmat{\alpha_{\lambda,n} & \beta_{\lambda,n} \\ \beta_{\lambda,n} & \delta_{\lambda,n}}\,=\,
 \pmat{\Upsilon^* \\ \Phi^*} \left( \hat\Vv_{\lambda,n}+O\bfa O^*\right)^{-1} \pmat{\Upsilon & \Phi}
 $$
Now let $t_n^2=\max(r^{-1}_{3n},r^{-1}_{3n-1},r^{-1}_{3n-2})$ if $r_{3n-1}> 0$ for all $n$ and $t_n^2=\max(r^{-1}_{3n},
r^{-1}_{3n-2})$ in the case where $r_{3n-1}=0$ for all $n$.
We find $\hat\Vv_{\lambda,n} \xrightarrow{\Oo(t_n)} [\EE(\omega_x-\lambda)^{-1}]^{-1} \one\,=\, h_\lambda \one$
 by Lemma~\ref{lem-averages}  where
 $h_\lambda$ is the harmonic mean as in \eqref{eq-def-h}.
 By Lemma~\ref{lem-ot-close} we therefore obtain for $\lambda\in I_{\Ab,\nu}$ that
 $$
 (\alpha_{\lambda,n},\beta_{\lambda,n},\beta_{\lambda,n}^{-1},\delta_{\lambda,n}) 
 \xrightarrow{\Oo(t_n)} (\alpha^\Ab_{\lambda},\beta^\Ab_{\lambda},(\beta^\Ab_{\lambda})^{-1},\delta^\Ab_{\lambda}) \;.
 $$
 Under the assumptions we have $\sum_{n=1}^\infty t_n^2<\infty$ and we can use Theorem~\ref{th-random} to see that in $I_{\Sb,\nu}$ there is almost surely spectrum for $\Hh_\omega$ and it is almost surely absolutely continuous.
\end{proof}

\section*{Acknowledgement}

This research  has received funding from the Chilean program CONICYT under grant Fondecyt Regular 1161651 and the European 
Union 7-th framework program FP7/2007-2013 under REA grant 291734 at the Institute of Science and Technology in Klosterneuburg, Austria where the work started.


\begin{thebibliography}{AAA}




\bibitem[ASW]{ASW} M. Aizenman, R. Sims and S. Warzel, {\sl Stability of the absolutely continuous spectrum of random
Schr\"odinger operators on tree graphs}, Prob. Theor. Rel. Fields, {\bf 136}, 363-394 (2006)



\bibitem[AW]{AW} M. Aizenman and S. Warzel,
{\sl Resonant delocalization for random Schrödinger operators on tree graphs},
J. Eur. Math. Soc., {\bf 15} (4), 1167-1222 (2013)












\bibitem[CL]{CL} R, Carmona and J. Lacroix, {\sl Spectral Theory of Random Schr\"odinger Operators},
Birkh\"auser, Boston, 1990








\bibitem[FHS]{FHS} R. Froese, D. Hasler and W. Spitzer,
{\sl Transfer matrices, hyperbolic geometry and absolutely continuous spectrum for some discrete Schr\"odinger operators on graphs}, 
J. Funct. Anal. {\bf 230}, 184-221 (2006)



\bibitem[FHS2]{FHS2} R. Froese, D. Hasler and W. Spitzer, 
{\sl Absolutely continuous spectrum for the Anderson Model on a tree:
A geometric proof of Klein's Theorem}, Commun. Math. Phys. {\bf 269}, 239-257 (2007)


 
\bibitem[KP]{KP} S. Kahn and D. Pearson, {\sl Subordinacy and spectral theory for infinite matrices},
Helv. Phys. Acta {\bf 65}, 505-527 (1992)



\bibitem[KLW]{KLW} M. Keller, D. Lenz and S. Warzel, 
{\sl On the spectral theory of trees with finite cone type}, Israel J. Math. {\bf 194}, 107-135 (2013)

\bibitem[KLW2]{KLW2} M. Keller, D. Lenz and S. Warzel, 
{\sl Absolutely continuous spectrum for random operators on trees of finite cone type}, J. D' Analyse Math. {\bf 118}, 363-396 (2012)



\bibitem[KiLS]{KLS} A. Kiselev, Y. Last and B. Simon
{\sl Modified Pr\"ufer and EFGP transforms and the spectral analysis of one-dimensional Schr\"odinger operators},
Comm. Math. Phys. {\bf 194},\,1-45 (1997)

 
\bibitem[Kl]{Kl} A. Klein, {\sl Extended states in the Anderson model on the Bethe lattice},
Advances in Math. {\bf 133}, 163-184 (1998)


\bibitem[KS]{KS} A. Klein and C. Sadel, {\sl Absolutely Continuous Spectrum for Random Schr\"odinger Operators
on the Bethe Strip}, Math. Nachr. {\bf 285}, 5-26 (2012)


\bibitem[LS]{LaSi} Y. Last and B. Simon,
{\sl Eigenfunctions, transfer matrices, and absolutely continuous spectrum of one-dimensional Schr\"odinger operators},
Inv. Math. {\bf 135}, 329-367 (1999)

\bibitem[MZ]{MZ} S. Molchanov and L. Zheng, {\sl Cluster expansion of the resolvent for the Schrödinger operator on non-percolating graphs with applications to Simon-Spencer type theorems and localization},
to be published in J. Spec. Th.

\bibitem[Sa1]{Sa} C. Sadel,
{\sl Absolutely continuous spectrum for random Schr\"odinger operators on tree-strips of finite cone type},
Annales Henri Poincar\'e, {\bf 14}, 737--773 (2013)

\bibitem[Sa2]{Sa2} C. Sadel,
{\sl Absolutely continuous spectrum for random Schr\"odinger operators on the Fibbonacci and similar tree-strips},
Math. Phys. Anal. Geom. {\bf 17},409--440 (2014)

\bibitem[Sa3]{Sa3} C. Sadel, 
{\sl Anderson transition at two-dimensional growth rate on antitrees and spectral theory for operators with one propagating channel}, 
Annales Henri Poincare {\bf 17}, 1631-1675 (2016)



\end{thebibliography}
\end{document}